\documentclass[12pt, draftclsnofoot, onecolumn]{IEEEtran}

\usepackage[utf8]{inputenc} 
\usepackage[T1]{fontenc}
\usepackage{url}              %
\usepackage{cite}             %

\usepackage[cmex10]{amsmath}  %
\usepackage{mleftright}       %
\mleftright                   %

\usepackage{graphicx}         %
\usepackage{booktabs}         %
\usepackage{multirow}

\usepackage[hidelinks]{hyperref}
\usepackage{textcomp}
\usepackage{arydshln}
\usepackage{amsthm,amsfonts,amssymb,amsmath}
\usepackage{setspace}
\usepackage[ruled,linesnumbered]{algorithm2e}
\usepackage{xcolor}

\newtheorem{theorem}{Theorem}
\newtheorem{corollary}{Corollary}
\newtheorem{lemma}{Lemma}
\newtheorem{remark}{Remark}

\DeclareMathOperator{\lcm}{lcm}

\makeatletter
\def\old@comma{,}
\catcode`\,=13
\def,{%
  \ifmmode%
    \old@comma\discretionary{}{}{}%
  \else%
    \old@comma%
  \fi%
}
\makeatother

\begin{document}

\allowdisplaybreaks
\title{The Asymptotic Capacity of $X$-Secure\\$T$-Private Linear Computation with\\Graph Based Replicated Storage}

\author{
    Haobo~Jia
    and Zhuqing~Jia,~\IEEEmembership{Member,~IEEE}
    \thanks{H. Jia and Z. Jia are with the School of Artificial Intelligence, Beijing University of Posts and Telecommunications, Beijing, 100086 China (e-mail: jiahaobo@bupt.edu.cn; zhuqingj@bupt.edu.cn).}
    \thanks{This work was presented in part at IEEE ISIT 2023 in \cite{doubleJia_XSTPLC}.}
}

\maketitle

\begin{abstract}
The problem of $X$-secure $T$-private linear computation with graph based replicated storage (GXSTPLC) is to enable the user to retrieve a linear combination of messages privately from a set of $N$ distributed servers where every message is only allowed to store among a subset of servers (in the form of secret shares) subject to an $X$-security constraint, i.e., any groups of up to $X$ colluding servers must reveal nothing about the messages. Besides, any groups of up to $T$ servers cannot learn anything about the coefficients of the linear combination retrieved by the user. In this work, we completely characterize the asymptotic capacity of GXSTPLC, i.e., the supremum of average number of desired symbols retrieved per downloaded symbol, in the limit as the number of messages $K$ approaches infinity. Specifically, it is shown that a prior linear programming based upper bound on the asymptotic capacity of GXSTPLC due to Jia and Jafar is tight (thus settles their conjecture) by constructing achievability schemes. Notably, our achievability scheme also settles the exact capacity (i.e., for finite $K$) of $X$-secure linear combination with graph based replicated storage (GXSLC). Our achievability proof builds upon an achievability scheme for a closely related problem named asymmetric $\mathbf{X}$-secure $\mathbf{T}$-private linear computation with graph based replicated storage (Asymm-GXSTPLC) that guarantees non-uniform security and privacy levels across messages and coefficients (of the desired linear combination). In particular, by carefully designing Asymm-GXSTPLC settings for GXSTPLC problems, the corresponding Asymm-GXSTPLC schemes can be reduced to asymptotic capacity achieving schemes for GXSTPLC. In regard to the achievability scheme for Asymm-GXSTPLC, interesting aspects of our construction include a novel query and answer design which makes use of a Vandermonde decomposition of Cauchy matrices, and a trade-off among message replication, security and privacy thresholds.
\end{abstract}

\section{Introduction}
Motivated by security and privacy concerns in prevalent distributed storage systems, private information retrieval (PIR) \cite{PIRfirstjournal,Yekhanin,Sun_Jafar_PIR} and private linear computation (PLC) \cite{Sun_Jafar_PC} are paradigms that allow a user to retrieve a message out of $K$ messages (a linear combination of $K$ messages) from a set of $N$ distributed servers without revealing any information about the index of the desired message (the coefficients of the desired linear combination) to any one of the servers. Recent interest in PIR and PLC from an information-theoretic perspective has successfully characterized a series of capacity results in terms of the download rate (i.e., the supremum of average number of desired symbols retrieved per downloaded symbol) for PIR and PLC under various constraints, e.g., 
PIR/PLC with $T$-privacy\cite{Sun_Jafar_TPIR,Sun_Jafar_PC,Mirmohseni_Maddah_ws,Lu_Jia_Jafar_DB,Raviv_Karpuk,Chen_Wang_Jafar_Search_trans}, 
PIR with arbitrary collusion pattern\cite{Tajeddine_Gnilke_Karpuk_Etal,Yao_Kang_Wei_APIR,Cheng_Kang_Wei_SAPIR}, 
PIR/PLC with coded storage\cite{Banawan_Ulukus,FREIJ_HOLLANTI,Sun_Jafar_MDSTPIR,Karpuk_David,Obead_Sarah_Kliewer,Tajeddine_Gnilke_Karpuk_Hollanti,Tajeddine_Rouayheb,Wang_Skoglund_MDSPIR,Obead_Sarah_PLC,Obead_Sarah_Lin_PPC}, 
multi-message PIR\cite{Banawan_Ulukus_MPIR_trans,wang2021private,Shariatpanahi_Siavoshani_Maddah_trans}, 
PIR with constrained storage\cite{Banawan_Karim_Arasli_HCPIR,Attia_Kumar_Tandon_trans,Zhu_Jin_scPIR}, 
PIR with side-information\cite{Tandon_CPIR_conf,Kadhe_Garcia_Heidarzadeh_Rouayheb_Sprintson_trans,Wei_Banawan_Ulukus,Wei_Banawan_Ulukus_Side_trans,Chen_Wang_Jafar_Side_trans,wei2019capacity}, 
multi-round PIR\cite{Sun_Jafar_MPIR,Yao2019MPIR}, 
PIR against eavesdroppers\cite{Wang_Sun_Skoglund_trans,Wang_Skoglund_SPIREve_conf,Banawan_Ulukus_Asymmetric_trans}, 
PIR/PLC with (MDS coded) $X$-secure storage\cite{Yu_Lagrange,Yang_Shin_Lee,Jia_Sun_Jafar_XSTPIR,Jia_Jafar_MDSXSTPIR,Jia_Jafar_GXSTPIR,Jia_Jafar_XSTPFSL,cheng2023asymptotic,doubleJia_XSTPLC}, 
symmetric secure PIR\cite{Sun_Jafar_SPIR,Zhu_Yan_Tang,Wang_Skoglund_SPIRAd_conf} 
and PIR/PLC with other models/constraints\cite{Banawan_Ulukus_Byzantine_trans,Banawan_Ulukus_Traffic_trans,Shah_Rashmi_Kannan,Zhang_Wang_Wei_Ge_trans,Sun_Jafar_PIRL,Jia_Sun_Jafar,Raviv_Tamo_Yaakobi_trans,Tian_Sun_Chen_PIR_trans,Sadeh_Bar_Gu,Banawan_ulukus_GPIR,keramaati2020private}.
One of the common yet restrictive assumptions in most information-theoretic PIR/PLC literature is the global data availability, i.e., each message is allowed to be stored at all of the $N$ servers in either plain replicated form or (securely) coded form, which could be cumbersome in some of the real-world applications.

In this work, we consider the problem of $X$-secure $T$-private linear computation with graph based replicated storage (GXSTPLC). In a graph based replicated storage model, the $K$ messages are partitioned into $M$ subsets and the messages of each subset are only allowed to `replicate' among a subset of servers (in either plain replicated form or certain securely coded form if data security is necessary). Recall that this allows for a (hyper)graph representation where each server is represented by a vertex and each message set is represented by a (hyper)edge which consists of the servers where the message set is allowed to `replicate', and this is referred to as {\it storage pattern}. There is no shortage of such scenarios where global data availability is not applicable. As shown in Figure \ref{fig:1} where multiple sources distribute their messages to the servers. However, due to a variety of constraints such as geographic blocking, network connectivity and data security, each source is only possible to distribute its messages to a subset of servers. Meanwhile, the user wishes to retrieve a linear combination of all of the $K$ messages without disclosing any information about the coefficients of the linear combination to any group of up to $T$ colluding servers. Besides, any groups of up to $X$ colluding servers are not allowed to reveal any information about the messages.

\begin{figure}[!h]
\centering
\includegraphics[scale=0.5]{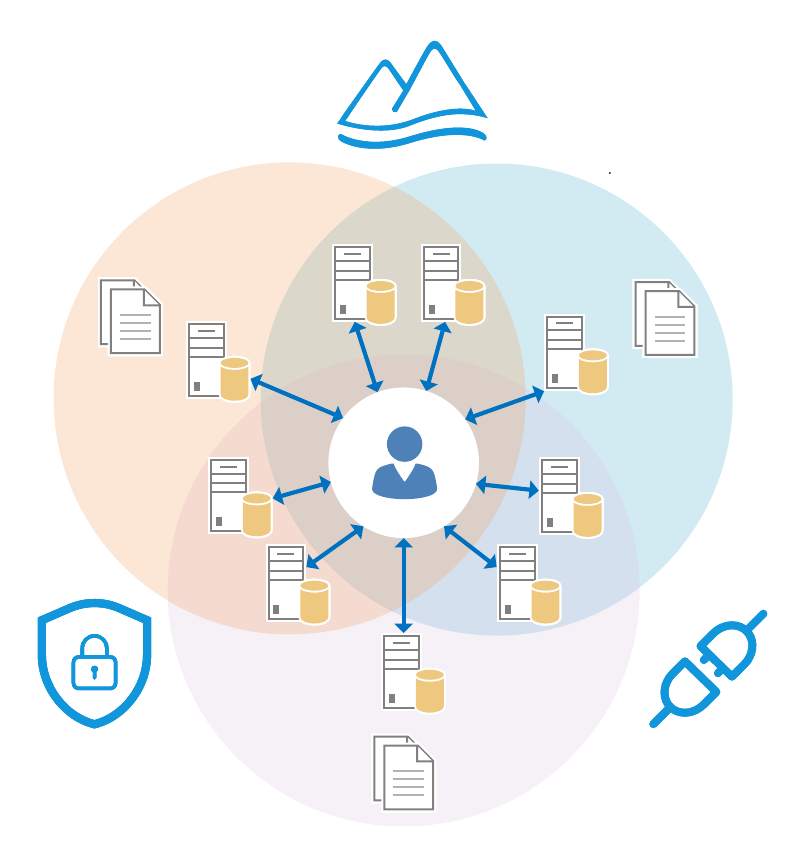}
\caption{Non-uniform data availability due to various constraints.}
\label{fig:1}
\end{figure}

There are several prior works on the problem of TPIR with graph based replicated storage (GTPIR), that is, a special information retrieval setting of our problem where $X=0$. For example, Raviv et al. \cite{Raviv_Tamo_Yaakobi_trans} proposed an achievability scheme that achieves the PIR rate of $R=1/N$ as long as each message is replicated at least $T+1$ times. Banawan et al. \cite{Banawan_ulukus_GPIR} considered a dual setting where $T=1$ and every server stores exactly $2$ messages. Sadeh et al. \cite{Sadeh_Bar_Gu} investigated upper and/or lower bounds on the capacity of GPIR (i.e., GTPIR where $T=1$) for several specific graph families by exploiting graph structures. Most relevant to this work is the problem of $X$-secure $T$-private information retrieval with graph based replicated storage (GXSTPIR) considered by Jia et al. \cite{Jia_Jafar_GXSTPIR}. In particular, a general achievability scheme for GXSTPIR is presented and the asymptotic capacity (i.e., for large $K$) of GTPIR is partially settled for the settings where each message is replicated no more than $T+2$ times. Indeed, the problem of GXSTPLC considered in this work can be viewed as a linear computation extension of GXSTPIR. However, the general capacity of both the two problems remains largely open.

The main contribution of this work is the complete asymptotic capacity characterization for the problem of GXSTPLC, also the problem of GXSTPIR as an immediate corollary (recall that in general PLC schemes can be used to construct PIR schemes while achieving the same rate). This also settles the conjecture on the asymptotic capacity of GXSTPIR in \cite{Jia_Jafar_GXSTPIR}, i.e., our result shows that the upper bound on the asymptotic capacity of GXSTPIR in \cite{Jia_Jafar_GXSTPIR} is tight by presenting matching achievability schemes. As pointed out in \cite{Jia_Jafar_GXSTPIR}, the characterization of the asymptotic capacity is quite meaningful for PIR/PLC problems because PIR/PLC capacities tend to converge to their asymptotic values exponentially (in the number of messages $K$) so that even for mildly large $K$ the gap is vanishingly small. Furthermore, our achievability scheme recovers many known asymptotic capacity results of various PIR/PLC problems as special cases including PIR/PLC, $T$-private PIR/TPLC, $X$-secure $T$-private PIR/PLC, thus can be viewed as a multidimensional generalization of these problems. The key to the achievability proof of the capacity of GXSTPLC turns out to be a closely related problem named \emph{asymmetric} $\mathbf{X}$-secure $\mathbf{T}$-private linear computation with graph based replicated storage (Asymm-GXSTPLC), where instead of the same privacy/security level on the coefficients/messages across message sets, non-uniform privacy and security thresholds are required by \emph{asymmetric} $\mathbf{X}$-security and $\mathbf{T}$-privacy (hence boldface $\mathbf{X}$ and $\mathbf{T}$ indicating tuples consisting of various security and privacy thresholds across message sets). Specifically, for any given GXSTPLC problem where the storage pattern is determined, we carefully construct a corresponding \emph{asymmetric} setting that can be reduced to the original GXSTPLC problem and achieves the same PLC rate. More importantly, we show that it is possible to achieve the rate that matches the asymptotic capacity of the GXSTPLC problem for the corresponding \emph{asymmetric} setting by constructing achievability schemes. The main technical novelty in this regard is that based on the idea of {\it cross subspace alignment} and a structure inspired by dual generalized Reed-Solomon (GRS) codes which are already developed as key components to the construction of a general achievability scheme of GXSTPIR in \cite{Jia_Jafar_GXSTPIR}, we introduce a novel construction of queries and answer symbols that is inspired by a Vandermonde decomposition of Cauchy matrices. This allows the alignment of the symbols of the desired linear combination {\it across message sets}, making it possible to retrieve the desired linear combination correctly. Besides, our achievability scheme for the asymmetric problem reveals an interesting trade-off among replication factors (i.e., the number of servers where the message set is allowed to store), security and privacy thresholds.

One interesting aspect that distinguishes our result from prior works is that although graph based replicated storage patterns naturally admit (hyper)graph representations, our achievability proof barely involves graph theory techniques or arguments. Instead, since the asymptotic capacity of GXSTPLC is presented in the form of a linear programming whose feasible region is a function of the storage pattern, our achievability proof depends mainly on the appropriate construction of the corresponding asymmetric setting, the achievability scheme for the asymmetric problem, and the property of the linear programming itself. It is hence of border interest to see if this methodology is instructive for finite $K$ settings. Another noteworthy result is the characterization of the exact (i.e., non-asymptotic) capacity of the problem of $X$-secure linear computation with graph based replicated storage (GXSLC), i.e., a special setting of GXSTPLC where $T=0, X\geq 1$ and the coefficients are non-zero. We refer the reader to Remark \ref{remark:gxslc} for details. 

The remainder of this paper is organized as follows. We formally define the problem of GXSTPLC in Section \ref{sec:ps}, and the main result is presented in Section \ref{sec:result}. Section \ref{sec:achi} is dedicated to the formal proof of our main result along with observations and illustrative examples. Finally, we conclude the paper in Section \ref{sec:conclusion}.

\noindent {\it Notations:} Bold symbols are used to denote vectors, matrices and tuples, while calligraphic symbols denote sets. By convention, let the empty product be the multiplicative identity, and the empty sum be the additive identity. $\mathbb{Z}_{\geq 0}$ denotes the set $\{0,1,2,\cdots\}$. For any positive integer $N$, $[N]$ denotes the index set $\{1,2,\cdots, N\}$. For an index set $\mathcal{I}\subset[N]$ and a set of random variables $\mathcal{X}=\{X_1, X_2, \cdots, X_N\}$, $\mathcal{X}_\mathcal{I}$ denotes the set $\{X_i \mid i\in\mathcal{I}\}$. For arbitrary $x\in\mathbb{R}$, $(x)^+$ denotes $\max(0,x)$.

\section{Problem Statement}\label{sec:ps}
\begin{figure}
\centering
\includegraphics[scale=0.62]{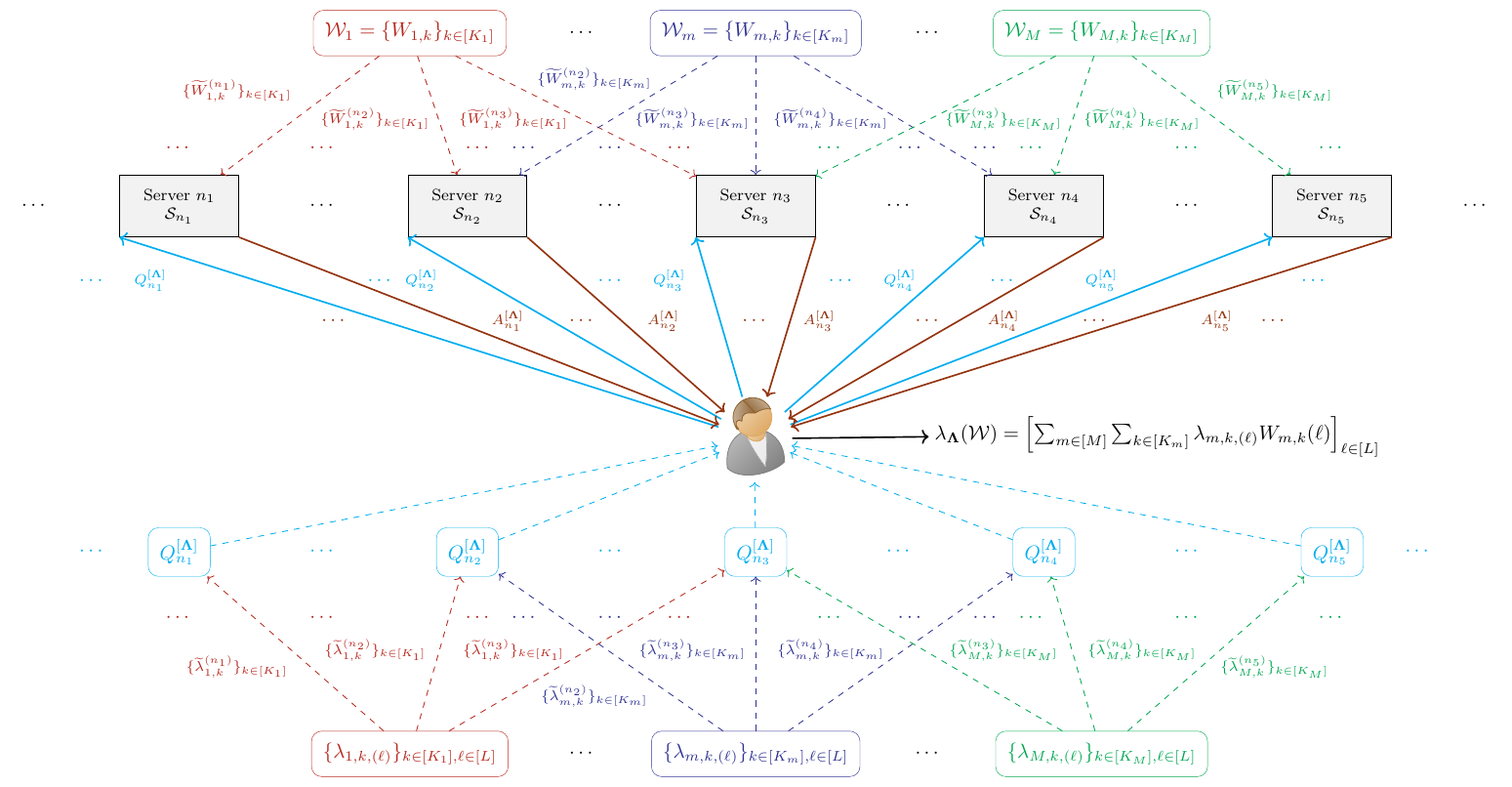}
\caption{The problem of $X$-secure $T$-private linear computation with graph based replicated storage.}
\label{fig:gxstplc}
\end{figure}
Consider the private linear computation problem as shown in Figure \ref{fig:gxstplc} with $N$ servers ($1,2,\cdots,N$) and $K$ independent messages. The set of the messages $\mathcal{W}$ is partitioned into $M$ disjoint subsets $\mathcal{W}_1,\mathcal{W}_2,\cdots,\mathcal{W}_M$, that is, $\mathcal{W}=\bigcup_{m\in[M]} \mathcal{W}_m, \mathcal{W}_i\cap\mathcal{W}_j=\emptyset,\ \forall i\neq j,\ i,j\in[M]$. For each $m\in[M]$, the subset $\mathcal{W}_m$ is comprised of $K_m$ messages, i.e., $\mathcal{W}_m=\{W_{m,1},W_{m,2},\cdots,W_{m,K_m}\}$. Each of the messages consists of $L$ i.i.d. uniform symbols from a finite field $\mathbb{F}_q$, and is represented as an $L$-dimensional vector, i.e., for all $m\in[m], k\in[K_m]$, $W_{m,k}=(W_{m,k}(1),W_{m,k}(2),\cdots, W_{m,k}(L))$, and
\begin{align}
    H((W_{m,k})_{m\in[M],k\in[K_m]})=KL,
\end{align}
in $q$-ary units. To characterize the storage pattern, i.e., the pattern that the $M$ message sets are `replicated' among the $N$ servers, let us define $M$ sets of servers, corresponding to the $M$ sets of messages, i.e., for all $m\in[M]$, we define
\begin{align}
    \mathcal{R}&=\{\mathcal{R}_1,\mathcal{R}_2,\cdots,\mathcal{R}_M\},\\
    \mathcal{R}_m&=\{\mathcal{R}_m(1),\mathcal{R}_m(2),\cdots,\mathcal{R}_m(\rho_m)\},0<\rho_m\le N,
\end{align}
where the messages $\mathcal{W}_m$ are allowed to `replicate' among the servers $\mathcal{R}_m$. $\rho_m$ is referred to as \emph{replication factor} of the $m^{th}$ message set.
Note that we can equivalently define the dual representation of the storage pattern, i.e., for all $n\in[N]$, Server $n$ is allowed to store a `replicated' copy of the messages $\mathcal{W}_m, m\in\mathcal{M}_n$, where $\mathcal{M}_n=\{m\in[M]\mid\mathcal{R}_m\ni n\}$.

Let us further elaborate on the above notations via the following example. Assume that we have $M=4$ message sets $\mathcal{W}_1, \mathcal{W}_2, \mathcal{W}_3, \mathcal{W}_4$ that are stored at $N=7$ servers as shown in the following table.
    \begin{center}
        \begin{tabular}{cccccccc}%
            \hline &Server $1$&Server $2$&Server $3$&Server $4$&Server $5$&Server $6$&Server $7$\\\hline 
            &$\mathcal{W}_1,\mathcal{W}_3$&$\mathcal{W}_1,\mathcal{W}_2,\mathcal{W}_4$&$\mathcal{W}_2,\mathcal{W}_4$&$\mathcal{W}_1,\mathcal{W}_2,\mathcal{W}_3$&$\mathcal{W}_2,\mathcal{W}_4$&$\mathcal{W}_2,\mathcal{W}_4$&$\mathcal{W}_3$\\\hline
        \end{tabular}
    \end{center}
For this example, we have
\begin{subequations}
\begin{align}
    \mathcal{R}_1&=\{1,2,4\}, &\rho_1&=3 &
    \mathcal{R}_2&=\{2,3,4,5,6\}, &\rho_2&=5\\
    \mathcal{R}_3&=\{1,4,7\}, &\rho_3&=3&
    \mathcal{R}_4&=\{2,3,5,6\}, &\rho_4&=4\\
    \mathcal{M}_1&=\{1,3\},&
    \mathcal{M}_2&=\{1,2,4\},&
    \mathcal{M}_3&=\{2,4\},&
    \mathcal{M}_4&=\{1,2,3\},\\
    \mathcal{M}_5&=\{2,4\},&
    \mathcal{M}_6&=\{2,4\},&
    \mathcal{M}_7&=\{3\}.
\end{align}
\end{subequations}

Note that due to security constraints, the messages are not replicated directly at the servers. Instead, they are replicated in the form of secret shares\footnote{Note that except graph based replicated storage pattern, security and recoverability constraints, we do not require any extra constraints on the storage. For example, we do not limit the storage overhead at each server, and the specific storage coding scheme is also considered as an object to be optimized to achieve the (asymptotic) capacity.}. In other words, the secret share of the message $W_{m,k}, m\in[M], k\in[K_m]$ stored at Server $n, n\in\mathcal{R}_m$, is denoted by $\widetilde{W}_{m,k}^{(n)}$, and the message $W_{m,k}$ must be recoverable from the collection of its secret shares,
\begin{align}
    H(W_{m,k}|\widetilde{W}^{(\mathcal{R}_m(1))}_{m,k},\cdots,\widetilde{W}^{(\mathcal{R}_m(\rho_m))}_{m,k})=0.
\end{align}
Besides, the secret shares of the messages must be independently generated, i.e., 
\begin{align}\label{eq:storind}
    H((\widetilde{W}_{m,k}^{(n)})_{m\in[M],k\in[K_m],n\in{\mathcal{R}_m}})=\sum_{m\in[M],k\in[K_m]}H((\widetilde{W}_{m,k}^{(n)})_{n\in\mathcal{R}_m}).
\end{align}
Therefore, the storage at Server $n, n\in[N]$, denoted as $\mathcal{S}_n$, is
\begin{align}\label{eq:defstor}
    \mathcal{S}_n=\{\widetilde{W}^{(n)}_{m,k} | \ m\in\mathcal{M}_n,\ k\in[K_m]\}.
\end{align}
The $X$-security constraint requires that any group of up to $X$ colluding servers, $0\leq X< N$, must reveal nothing about the messages, i.e., 
\begin{align}\label{eq:xsecure}
    I(\mathcal{S}_\mathcal{X};\mathcal{W})=0,\ \forall \mathcal{X}\subset[N],\ |\mathcal{X}|=X.
\end{align}
The user wishes to retrieve a linear combination of the $K$ messages without revealing any information about the coefficients of the linear combination to any group of up to $T$ colluding servers, $1\leq T<N$. In this work, we assume that the user is interested in the linear combination of the following form\footnote{There is a subtle difference between the form of linear combination considered in this work and that in prior works, e.g., \cite{Sun_Jafar_PC,Mirmohseni_Maddah_ws,Obead_Sarah_PLC,Obead_Sarah_Kliewer,cheng2023asymptotic}, where the desired linear combination is of the form $\sum_{m\in[M],k\in{K_m}} \lambda_{m,k}'W_{m,k}$. It is obvious that the prior form can be recovered as a special case of our form by setting $\lambda_{m,k,(\ell)}=\lambda_{m,k}'$ for all $m\in[M],k\in{K_m},\ell\in[L]$.}.
\begin{align}\label{eq:lincomb}
\lambda_{\boldsymbol{\Lambda}}(\mathcal{W})\triangleq \left[\sum_{m\in[M]}\sum_{k\in[K_m]}\lambda_{m,k,(\ell)}W_{m,k}(\ell)\right]_{\ell\in[L]},
\end{align}
where $\boldsymbol{\Lambda}=(\lambda_{m,k,(\ell)})_{m\in[M],k\in[K_m],\ell\in[L]}$ and $\lambda_{m,k,(\ell)}\in\mathbb{F}_q, m\in[M],k\in[K_m],\ell\in[L]$ are uniformly i.i.d. coefficients generated by the user privately. To this end, the user generates queries/secret shares of the coefficients of the linear combination of each message and sends them to the corresponding servers, i.e., the servers that store the secret shares of the corresponding message\footnote{It turns out that forcing the structure of queries to be in accordance with the storage pattern does not hurt the asymptotic capacity. Furthermore, this symmetry in the definition between the storage and queries helps us to identify an interesting connection between the problem of GXSTPLC and Asymm-GXSTPLC, which lies at the heart of our achievability proof.}. In particular, denote the query/share of the coefficients of the $k^{th}$ message in the $m^{th}$ message set for Server $n,n\in \mathcal{R}_m$ as $\widetilde{\lambda}_{m,k}^{(n)}$, the query sent to Server $n$, denoted as $Q_n^{[\boldsymbol{\Lambda}]}$, is
\begin{align}
    Q_n^{[\boldsymbol{\Lambda}]}=\left\{\widetilde{\lambda}_{m,k}^{(n)}\mid m\in\mathcal{M}_n, k\in[K_m]\right\},
\end{align}
and $T$-privacy constraint requires
\begin{align}
    I(Q^{[\boldsymbol{\Lambda}]}_{\mathcal{T}};\boldsymbol{\Lambda})=0,\ \forall \mathcal{T}\subset[N],\ |\mathcal{T}|=T.
\end{align} 
The coefficients must be able to be recovered from the collection of its queries/shares, i.e., for all $m\in[M], k\in[K_m]$, 
\begin{align}
    H\left(\lambda_{m,k,(1)}, \lambda_{m,k,(2)}, \cdots, \lambda_{m,k,(L)}\mid \widetilde{\lambda}_{m,k}^{(\mathcal{R}_m(1))}, \widetilde{\lambda}_{m,k}^{(\mathcal{R}_m(2))}, \cdots, \widetilde{\lambda}_{m,k}^{(\mathcal{R}_m(\rho_m))}\right) = 0.
\end{align}
The queries/shares of the coefficients of each message set must be generated independently.
\begin{align}\label{eq:queryind}
    H\left((\widetilde{\lambda}_{m,k}^{(n)})_{m\in[M], k\in[K_m], n\in\mathcal{R}_m}\right) = \sum_{m\in[M], k\in[K_m]} H\left((\widetilde{\lambda}_{m,k}^{(n)})_{n\in\mathcal{R}_m}\right).
\end{align}
Besides, the user has no prior knowledge of messages and server storage.
\begin{align}
    I(\mathcal{S}_{[N]};\boldsymbol{\Lambda},Q^{[\boldsymbol{\Lambda}]}_1,\cdots,Q^{[\boldsymbol{\Lambda}]}_N)=0.
\end{align}

Upon receiving $Q^{[\boldsymbol{\Lambda}]}_n$, Server $n$, $n\in[N]$ generates an answer string $A^{[\boldsymbol{\Lambda}]}_n$, which is a function of $Q^{[\boldsymbol{\Lambda}]}_n$ and its storage $\mathcal{S}_n$,
\begin{align}
    H(A^{[\boldsymbol{\Lambda}]}_n|Q^{[\boldsymbol{\Lambda}]}_n,\mathcal{S}_n)=0,
\end{align}
and returns the answer string $A^{[\boldsymbol{\Lambda}]}_n$ to the user.
The correctness constraint requires that the user must be able to recover the desired linear combination from the answers, i.e.,
\begin{align}
    H(\lambda_{\boldsymbol{\Lambda}}(\mathcal{W})|A^{[\boldsymbol{\Lambda}]}_{[N]},Q^{[\boldsymbol{\Lambda}]}_{[N]},\boldsymbol{\Lambda})=0.
\end{align}
The rate of a GXSTPLC scheme $R$ is the average number of $q$-ary symbols of the desired linear combination retrieved per downloaded $q$-ary symbol, that is, $R\triangleq L/D$, where $D$ is the expected total number of $q$-ary symbols downloaded by the user from all servers. The capacity of GXSTPLC, denoted as $C(N,X,T,\mathcal{W},\mathcal{R})$, is the supremum of $R$ over all feasible schemes. The asymptotic capacity of GXSTPLC is defined as $C_{\infty}=\lim_{K_1,K_2,\cdots,K_M\rightarrow\infty} C(N,X,T,\mathcal{W},\mathcal{R}).$

\section{Main Result: The Asymptotic Capacity of GXSTPLC}\label{sec:result}
The main result of this work is the complete characterization of the asymptotic capacity of GXSTPLC, as formally stated in the following theorem.
\begin{theorem}\label{thm:main}
    The asymptotic capacity of GXSTPLC (as a function of storage pattern $\mathcal{R}$) where $0\leq X<N$ and $1\leq T<N$ is
    \begin{align}
        C_\infty(\mathcal{R})=\left\{
        \begin{array}{ll}
        0,&\min_{m\in[M]}\rho_{m}\leq X+T\\
        \max_{(D_1, \cdots, D_N)\in\mathcal{D}}~~~ \left(\sum_{n\in[N]}D_n\right)^{-1}, &\min_{m\in[M]}\rho_{m}> X+T
        \end{array}
        \right.\label{eq:lp}
    \end{align}
    where $\mathcal{D}$ is defined as
    \begin{align}
        \mathcal{D}&\triangleq\left\{(D_1, \cdots, D_N)\in\mathbb{R}_+^N~\Big|~ \sum_{n\in \mathcal{R}_m'} D_n\geq 1, ~\forall m\in[M],  \mathcal{R}_m'\subset\mathcal{R}_m, |\mathcal{R}_m'|=|\mathcal{R}_m|-X-T \right\}\label{eq:defD}.
    \end{align}
\end{theorem} 
The achievability proof of Theorem \ref{thm:main} is presented in Section \ref{sec:achi}. The converse follows from a \emph{reductio ad absurdum} argument as follows.
\begin{proof}{(Converse of Theorem \ref{thm:main})} To set up a proof by contradiction, let us assume that there exists an achievability scheme for a given storage pattern $\mathcal{R}$ that achieves a PLC rate $R>C_\infty(\mathcal{R})$. Now from this let us construct an achievability scheme for the problem of $X$-secure $T$-private information retrieval with graph based replicated storage (GXSTPIR) by restricting the coefficients of the linear combination $\boldsymbol{\Lambda}$ to the set $\mathcal{E}=\{\boldsymbol{\Lambda}_{m,k}\}_{m\in[M],k\in[K_m]}$ where $\boldsymbol{\Lambda}_{m',k'}=(\lambda_{m,k,(\ell)}^{[m',k']})_{m\in[M],k\in[K_m],\ell\in[L]}$ and $\lambda_{m,k,(\ell)}^{[m',k']}=1$ if and only if $m=m', k=k'$ otherwise $\lambda_{m,k,(\ell)}^{[m',k']}=0$. In other words, if the $k^{th}$ message in the $m^{th}$ message set is desired by the user, the user set $\boldsymbol{\Lambda}=\boldsymbol{\Lambda}_{m,k}$. Note that conditioning on $\boldsymbol{\Lambda}$ taking values over the set $\mathcal{E}$ does not change the rate of the PLC scheme $R$ since $\boldsymbol{\Lambda}$ must be independent of the query $Q_n^{[\boldsymbol{\Lambda}]}$ sent to Server $n$ and thus the answer $A_n^{[\boldsymbol{\Lambda}]}$ returned by Server $n$. Therefore, the resulting GXSTPIR scheme achieves the rate $R>C_\infty(\mathcal{R})$ for the same storage pattern $\mathcal{R}$ which leads to a contradiction that the achievable rate of GXSTPIR is bounded above by $C_\infty(\mathcal{R})$ (see \cite{Jia_Jafar_GXSTPIR}). This completes the proof.
\end{proof}
From the converse proof of Theorem \ref{thm:main}, we have the following corollary.
\begin{corollary}\label{col:gxstpir}
    The asymptotic capacity of GXSTPIR (as a function of storage pattern $\mathcal{R}$) where $0\leq X<N$ and $1\leq T<N$ is \eqref{eq:lp}.
\end{corollary}

\section{Achievability Proof of Theorem \ref{thm:main}}\label{sec:achi}
\subsection{Asymmetric $\mathbf{X}-$Secure $\mathbf{T}$-Private Linear Computation with Graph Based Replicated Storage}
In this section, we present an achievability scheme for the problem of asymmetric $\mathbf{X}$-secure and $\mathbf{T}$-private linear computation with graph based replicated storage that turns out to be the key ingredient to the asymptotic capacity achieving scheme of GXSTPLC. The problem of Asymm-GXSTPLC is to allow the user retrieve a linear combination of messages in the form of \eqref{eq:lincomb} from $N$ distributed servers that store the $M$ message sets according to a graph based replicated storage (following the same definitions and notations in Section \ref{sec:ps}) while guaranteeing an \emph{asymmetric} $\mathbf{T}$-privacy constraint, i.e., for all $m\in[M]$, any up to $T_m$ colluding servers must not reveal anything about the coefficients of the linear combination for the $m^{th}$ message set, $1\leq T_m<N$. Formally, for all $m\in[M]$,
\begin{align}
    I(Q^{[\boldsymbol{\Lambda}]}_{\mathcal{T}_m};(\lambda_{m,k,(\ell)})_{k\in[K_m],\ell\in[L]})=0,\ \forall \mathcal{T}_m\subset[N],\ |\mathcal{T}_m|=T_m,
\end{align} 
and $\mathbf{T}=(T_1, T_2, \cdots, T_m)$. 
Besides, for all $m\in[M]$, any up to $X_m$ colluding servers cannot learn anything about the messages in the $m^{th}$ message set, $0\leq X_m<N$, i.e.,
\begin{align}
    I(\mathcal{S}_{\mathcal{X}_m};\mathcal{W}_m)=0,\ \forall \mathcal{X}_m\subset[N],\ |\mathcal{X}_m|=X_m,
\end{align}
and $\mathbf{X}=(X_1, X_2, \cdots, X_m)$.

One important observation is that it is possible to construct GXSTPLC schemes from Asymm-GXSTPLC schemes. A naive idea in this regard is to apply $T_1=T_2=\cdots=T_m=T$ and $X_1=X_2=\cdots=X_m=X$ since it is easy to identify that Asymm-GXSTPLC reduces to GXSTPLC in this case. However, our idea goes one step further by exploiting finer \emph{asymmetric} privacy and security structures and finally leads to the complete characterization of the asymptotic capacity of GXSTPLC. Let us elaborate via an illustrative example. Consider the following example where we have $N=7$ servers that store $M=4$ message sets according to the following storage pattern (this is exactly the example in Section \ref{sec:ps}).
\begin{subequations}\label{eq:storageex1}
\begin{align}
    \mathcal{R}_1&=\{1,2,4\},&
    \mathcal{R}_2&=\{2,3,4,5,6\},\\
    \mathcal{R}_3&=\{1,4,7\},&
    \mathcal{R}_4&=\{2,3,5,6\}.
\end{align}
\end{subequations}
Conversely, we can also write
\begin{subequations}
\begin{align}
    \mathcal{M}_1&=\{1,3\},&
    \mathcal{M}_2&=\{1,2,4\},\\
    \mathcal{M}_3&=\{2,4\},&
    \mathcal{M}_4&=\{1,2,3\},\\
    \mathcal{M}_5&=\{2,4\},&
    \mathcal{M}_6&=\{2,4\},\\
    \mathcal{M}_7&=\{3\}.
\end{align}
\end{subequations}
Besides, for asymmetric security and privacy thresholds, let us set $X_1=X_2=X_3=X_4=0$, and $T_1=T_3=1, T_2=T_4=2$. Assume that there exists an Asymm-GXSTPLC scheme for this setting that achieves the rate $R$. If we merge Server $2$ and Server $3$ into one new server (denoted as Server $2\cup3$) by assigning the storage and queries for the two servers to the new server, and similarly merge Server $5$ and Server $6$ into one new server (denoted as Server $5\cup6$), we indeed end up with a linear computation scheme for the following storage pattern with $N=5$ servers (listed as $(1,2\cup3,4,5\cup6,7)$).
\begin{subequations}\label{eq:storageex1reduced}
\begin{align}
    \mathcal{M}_1&=\{1,3\},\\
    \mathcal{M}_{2\cup3}&=\mathcal{M}_{2}\cup\mathcal{M}_{3}=\{1,2,4\},\\
    \mathcal{M}_4&=\{1,2,3\},\\
    \mathcal{M}_{5\cup6}&=\mathcal{M}_{5}\cup\mathcal{M}_{6}=\{2,4\},\\
    \mathcal{M}_7&=\{3\}.
\end{align}
\end{subequations}
Now let us see to what extent is this $5$-server linear computation scheme private. Observing the fact that the coefficients of the linear combination for each message set must be independently secret shared according to \eqref{eq:queryind}, Server $2\cup3$ as well as Server $5\cup6$ learns nothing about the coefficients of the linear combination for the $2^{nd}$ and $4^{th}$ message sets even if they see two queries of the coefficients for the two message sets that were originally designated for Server $2,3$ and $5,6$ respectively. This is because we set asymmetric privacy thresholds $T_2=T_4=2$. Besides, Server $2\cup 3$ cannot reveal anything about the coefficients of the linear combination for the $1^{st}$ message set because it can see only one query of the coefficients for the first message set and $T_1=1$. Accordingly, it is easy to check that the coefficients of the linear combination for each message set are kept private from each of the $5$ servers. Therefore this $N=5$-server linear computation scheme is indeed $T=1$-private while achieving the rate $R$. According to Theorem \ref{thm:main}, the rate of GXSTPLC with $X=0, T=1$ and the storage pattern in \eqref{eq:storageex1reduced} is bounded above by $2/7$. Therefore, if it is possible to achieve the rate $R=2/7$ for the asymmetric setting in the example, we immediately settle the asymptotic capacity of the corresponding GXSTPLC problem. This is shown to be true as an example of a more general achievability result (i.e., Lemma \ref{lemma:asymm}) in Section \ref{sec:asymmex1}. 

In general, in the achievability proof of Theorem \ref{thm:main}, for any given GXSTPLC setting, i.e., for any given set of storage pattern and privacy and security thresholds, we first carefully construct a corresponding Asymm-GXSTPLC setting which we refer to as \emph{augmented system} by identifying its storage pattern and \emph{asymmetric} privacy and security thresholds such that it can be leveraged to construct achievability schemes for the original GXSTPLC problem by merging servers. Then we conclude the achievability proof by showing that for the constructed augmented system, it is indeed possible to achieve the rate that equals the asymptotic capacity of the corresponding GXSTPLC problem. Two main technical challenges in this regard are: 1) such augmented system is not unique, therefore we must construct augmented systems that can be eventually exploited to construct asymptotic capacity-achieving GXSTPLC schemes; 2) the achievability scheme of Asymm-GXSTPLC itself requires interference alignment of various noise and interference symbols due to heterogeneous storage structure and asymmetric privacy and security levels across message sets, while guaranteeing that the desired linear combination is still retrievable. The first challenge is tackled in Section \ref{sec:proofmain}, and the second challenge is resolved as an achievability result for the problem of Asymm-GXSTPLC. This crucial achievability result is presented in the following lemma, and the remainder of this section is devoted to the proof of the achievability result.
\begin{lemma}\label{lemma:asymm}
    The following rate $R$ is achievable for Asymm-GXSTPLC. 
    \begin{align}
        R=\frac{\min_{m\in[M]}(\rho_m-X_m-T_m)^+}{N}.
    \end{align}
\end{lemma}
\begin{remark}
    Lemma \ref{lemma:asymm} reveals an interesting trade-off among replication factor $\rho_m$, security threshold $X_m$ and privacy threshold $T_m$. Specifically, in our achievability scheme, it is the minimum number of $(\rho_m-X_m-T_m)$ over $m\in[M]$ that determines the achievable rate. In other words, it is still possible to achieve a higher rate even if certain message sets are less replicated as long as their security and privacy requirements are also mild.
\end{remark}

\subsubsection{Motivating Example 1}\label{sec:asymmex1}
Continuing our discussion above, let us consider the motivating example with the storage pattern in \eqref{eq:storageex1} and the asymmetric security and privacy thresholds $X_1=X_2=X_3=X_4=0$, $T_1=T_3=1, T_2=T_4=2$. In particular, let us see how to achieve the rate $R=2/7$ for this Asymm-GXSTPLC setting that matches Lemma \ref{lemma:asymm} and the asymptotic capacity of the corresponding GXSTPLC setting. Recall that the $K$ messages are partitioned into four disjoint sets $\mathcal{W}_1$, $\mathcal{W}_2$, $\mathcal{W}_3$ and $\mathcal{W}_4$, where 
\begin{subequations}
    \begin{align}
        \mathcal{W}_1=\{W_{1,1},W_{1,2},\ldots,W_{1,K_1}\},\\
        \mathcal{W}_2=\{W_{2,1},W_{2,2},\ldots,W_{2,K_2}\},\\
        \mathcal{W}_3=\{W_{3,1},W_{3,2},\ldots,W_{3,K_3}\},\\
        \mathcal{W}_4=\{W_{4,1},W_{4,2},\ldots,W_{4,K_4}\}.
    \end{align}
\end{subequations}
Let us set $L=2$, i.e., each message is comprised of $L=2$ symbols from $\mathbb{F}_q, q\geq 9$,
\begin{align}
    W_{m,k}=(W_{m,k}(1),W_{m,k}(2)),\ \forall m\in[4],k\in[K_m].
\end{align}
Let $\alpha_1,\alpha_2,\alpha_3,\alpha_4,\alpha_5,\alpha_7,f_1,f_2$ be $9$ distinct constants from the finite field $\mathbb{F}_q$. For all $m\in[4],\ell\in\{1,2\},n\in\mathcal{R}_m$, let us define the following constants.
\begin{align}
    v_{n,m}&=\prod_{n'\in\mathcal{R}_m\setminus\{n\}}{(\alpha_n-\alpha_{n'})^{-1}},\\
    u_{m,l}&=\prod_{n\in\mathcal{R}_m}{(f_l-\alpha_n)}.
\end{align}
The storage at the $7$ servers, i.e. $\mathcal{S}_1,\mathcal{S}_2,\ldots,\mathcal{S}_7$ are constructed as follows.
\begin{subequations}
    \begin{align}
        \mathcal{S}_1&=\{\mathbf{W}^{(1)}_1,\mathbf{W}^{(1)}_3\},\ &
        \mathcal{S}_2&=\{\mathbf{W}^{(2)}_1,\mathbf{W}^{(2)}_2,\mathbf{W}^{(2)}_4\},\\
        \mathcal{S}_3&=\{\mathbf{W}^{(3)}_2,\mathbf{W}^{(3)}_4\},\ &
        \mathcal{S}_4&=\{\mathbf{W}^{(4)}_1,\mathbf{W}^{(4)}_2,\mathbf{W}^{(4)}_3\},\\
        \mathcal{S}_5&=\{\mathbf{W}^{(5)}_2,\mathbf{W}^{(5)}_4\},\ &
        \mathcal{S}_6&=\{\mathbf{W}^{(6)}_2,\mathbf{W}^{(6)}_4\},\\
        \mathcal{S}_7&=\{\mathbf{W}^{(7)}_3\},
    \end{align}
\end{subequations}
where for all $n\in\{1,2,4\}$,
\begin{align}
    \mathbf{W}^{(n)}_1=
    \begin{bmatrix}
        \frac{1}{\alpha_n-f_1}\mathbf{W}_{1,(1)}\\
        \frac{1}{\alpha_n-f_2}\mathbf{W}_{1,(2)},
    \end{bmatrix},
\end{align}
for all $n\in\{2,3,4,5,6\}$,
\begin{align}
    \mathbf{W}^{(n)}_2=
    \begin{bmatrix}
        \frac{1}{\alpha_n-f_1}\mathbf{W}_{2,(1)}\\
        \frac{1}{\alpha_n-f_2}\mathbf{W}_{2,(2)}
    \end{bmatrix},
\end{align}
for all $n\in\{1,4,7\}$,
\begin{align}
    \mathbf{W}^{(n)}_3=
    \begin{bmatrix}
        \frac{1}{\alpha_n-f_1}\mathbf{W}_{3,(1)}\\
        \frac{1}{\alpha_n-f_2}\mathbf{W}_{3,(2)},
    \end{bmatrix},
\end{align}
for all $n\in\{2,3,5,6\}$,
\begin{align}
    \mathbf{W}^{(n)}_4=
    \begin{bmatrix}
        \frac{1}{\alpha_n-f_1}\mathbf{W}_{4,(1)}\\
        \frac{1}{\alpha_n-f_2}\mathbf{W}_{4,(2)}
    \end{bmatrix},
\end{align}
and for all $m\in[4],\ell\in[2]$,
\begin{align}
    \mathbf{W}_{m,(\ell)}=\left[W_{m,1}(\ell),\cdots,W_{m,k_m}(\ell)\right]^\mathsf{T}.
\end{align}
To retrieve the desired linear combination of messages, the user generates i.i.d. uniform column vectors $\mathbf{Z}_{m,(\ell)}\in\mathbb{F}_q^{K_m}$ for all $m\in[4],\ell\in\{1,2\}$ and $\mathbf{Z}'_{m,(\ell)}\in\mathbb{F}_q^{K_m}$ for all $m\in\{2,4\},\ell\in\{1,2\}$ privately that are independent of $\boldsymbol{\Lambda}$.
The queries $Q_{[N]}^{[\boldsymbol{\Lambda}]}$ are constructed as follows.
\begin{subequations}
    \begin{align}
        Q_1^{[\boldsymbol{\Lambda}]}&=
        \{\mathbf{Q}^{[\boldsymbol{\Lambda}]}_{1,1},
        \mathbf{Q}^{[\boldsymbol{\Lambda}]}_{1,3}\},&\
        Q_2^{[\boldsymbol{\Lambda}]}&=
        \{\mathbf{Q}^{[\boldsymbol{\Lambda}]}_{2,1},
        \mathbf{Q}^{[\boldsymbol{\Lambda}]}_{2,2},
        \mathbf{Q}^{[\boldsymbol{\Lambda}]}_{2,4}\},\\
        Q_3^{[\boldsymbol{\Lambda}]}&=
        \{\mathbf{Q}^{[\boldsymbol{\Lambda}]}_{3,2},
        \mathbf{Q}^{[\boldsymbol{\Lambda}]}_{3,4}\},&\
        Q_4^{[\boldsymbol{\Lambda}]}&=
        \{\mathbf{Q}^{[\boldsymbol{\Lambda}]}_{4,1},
        \mathbf{Q}^{[\boldsymbol{\Lambda}]}_{4,2},
        \mathbf{Q}^{[\boldsymbol{\Lambda}]}_{4,3}\},\\
        Q_5^{[\boldsymbol{\Lambda}]}&=
        \{\mathbf{Q}^{[\boldsymbol{\Lambda}]}_{5,2},
        \mathbf{Q}^{[\boldsymbol{\Lambda}]}_{5,4}\},&\
        Q_6^{[\boldsymbol{\Lambda}]}&=
        \{\mathbf{Q}^{[\boldsymbol{\Lambda}]}_{6,2},
        \mathbf{Q}^{[\boldsymbol{\Lambda}]}_{6,4}\},\\
        Q_7^{[\boldsymbol{\Lambda}]}&=
        \{\mathbf{Q}^{[\boldsymbol{\Lambda}]}_{7,3}\},
\end{align}
\end{subequations}
where for all $n\in\{1,2,4\}$,
\begin{align}
    \mathbf{Q}^{[\boldsymbol{\Lambda}]}_{n,1}=
    \begin{bmatrix}
        u_{1,1}\boldsymbol{\lambda}_{1,(1)}+(\alpha_n-f_1)\mathbf{Z}_{1,(1)}\\
        u_{1,2}\boldsymbol{\lambda}_{1,(2)}+(\alpha_n-f_2)\mathbf{Z}_{1,(2)}
    \end{bmatrix},
\end{align}
for all $n\in\{2,3,4,5,6\}$,
\begin{align}
    \mathbf{Q}^{[\boldsymbol{\Lambda}]}_{n,2}=
    \begin{bmatrix}
        u_{2,1}\boldsymbol{\lambda}_{2,(1)}+(\alpha_n-f_1)(\mathbf{Z}_{2,(1)}+\alpha_n\mathbf{Z}'_{2,(1)})\\
        u_{2,2}\boldsymbol{\lambda}_{2,(2)}+(\alpha_n-f_2)(\mathbf{Z}_{2,(2)}+\alpha_n\mathbf{Z}'_{2,(2)})
    \end{bmatrix},
\end{align}
for all $n\in\{1,4,7\}$,
\begin{align}
    \mathbf{Q}^{[\boldsymbol{\Lambda}]}_{n,3}=
    \begin{bmatrix}
        u_{3,1}\boldsymbol{\lambda}_{3,(1)}+(\alpha_n-f_1)\mathbf{Z}_{3,(1)}\\
        u_{3,2}\boldsymbol{\lambda}_{3,(2)}+(\alpha_n-f_2)\mathbf{Z}_{3,(2)}
    \end{bmatrix},
\end{align}
for all $n\in\{2,3,5,6\}$,
\begin{align}
    \mathbf{Q}^{[\boldsymbol{\Lambda}]}_{n,4}=
    \begin{bmatrix}
        u_{4,1}\boldsymbol{\lambda}_{4,(1)}+(\alpha_n-f_1)(\mathbf{Z}_{4,(1)}+\alpha_n\mathbf{Z}'_{4,(1)})\\
        u_{4,2}\boldsymbol{\lambda}_{4,(2)}+(\alpha_n-f_2)(\mathbf{Z}_{4,(2)}+\alpha_n\mathbf{Z}'_{4,(2)})
    \end{bmatrix},
\end{align}
and for all $m\in[4],\ell\in[2]$,
\begin{align}
    \boldsymbol{\lambda}_{m,(\ell)}=[\lambda_{m,1,(\ell)},\cdots,\lambda_{m,K_m,(\ell)}]^\mathsf{T}.
\end{align}
The $T_1=T_3=1$ and $T_2=T_4=2$ privacy constraints are guaranteed by the (MDS coded) random noise vectors protecting the coefficients of the desired linear combination. Note that the desired linear combination is equivalently written in terms of $(\mathbf{W}_{m,(\ell)})_{m\in[4],\ell\in[2]}$ and $(\boldsymbol{\lambda}_{m,(\ell)})_{m\in[4],\ell\in[2]}$ as follows.
\begin{align}
    \lambda_{\boldsymbol{\Lambda}}(\mathcal{W})\hspace{-0.1cm}&=\hspace{-0.1cm}
    \begin{bmatrix}
        \sum_{m\in[4]}\sum_{k\in[K_m]} \lambda_{m,k,(1)}W_{m,k}(1)\\
        \sum_{m\in[4]}\sum_{k\in[K_m]} \lambda_{m,k,(2)}W_{m,k}(2)
    \end{bmatrix}\\&=
    \begin{bmatrix}
        \sum_{m\in[4]}\mathbf{W}_{m,(1)}^\mathsf{T}\boldsymbol{\lambda}_{m,(1)}\\
        \sum_{m\in[4]}\mathbf{W}_{m,(2)}^\mathsf{T}\boldsymbol{\lambda}_{m,(2)}
    \end{bmatrix}.
\end{align}

The query $Q_n^{[\boldsymbol{\Lambda}]}$ is sent to Server $n$ for all $n\in[7]$. Upon receiving the query $Q_n^{[\boldsymbol{\Lambda}]}$, Server $1$ responds to the user with the following answer.
\begin{align}
    A^{[\boldsymbol{\Lambda}]}_1
    =&v_{1,1}(\mathbf{W}^{(1)}_1)^\mathsf{T}\mathbf{Q}^{[\boldsymbol{\Lambda}]}_{1,1}+v_{1,3}(\mathbf{W}^{(1)}_3)^\mathsf{T}\mathbf{Q}^{[\boldsymbol{\Lambda}]}_{1,3}\\ 
    =&\frac{v_{1,1}u_{1,1}}{\alpha_1-f_1}\mathbf{W}_{1,(1)}^\mathsf{T}\boldsymbol{\lambda}_{1,(1)}+\frac{v_{1,1}u_{1,2}}{\alpha_1-f_2}\mathbf{W}_{1,(2)}^\mathsf{T}\boldsymbol{\lambda}_{1,(2)}+\frac{v_{1,3}u_{3,1}}{\alpha_1-f_1}\mathbf{W}_{3,(1)}^\mathsf{T}\boldsymbol{\lambda}_{3,(1)}\notag\\
    &+\frac{v_{1,3}u_{3,2}}{\alpha_1-f_2}\mathbf{W}_{3,(2)}^\mathsf{T}\boldsymbol{\lambda}_{3,(2)}+v_{1,1}I_1+v_{1,3}I_3,
\end{align}
where
\begin{align}
    I_1=&\mathbf{W}_{1,(1)}^\mathsf{T}\mathbf{Z}_{1,(1)}+\mathbf{W}_{1,(2)}^\mathsf{T}\mathbf{Z}_{1,(2)}\\
    I_3=&\mathbf{W}_{3,(1)}^\mathsf{T}\mathbf{Z}_{3,(1)}+\mathbf{W}_{3,(2)}^\mathsf{T}\mathbf{Z}_{3,(2)}.
\end{align}
The answer returned by Server $2$ is constructed as follows.
\begin{align}
    A^{[\boldsymbol{\Lambda}]}_2
    =&v_{2,1}(\mathbf{W}^{(2)}_1)^\mathsf{T}\mathbf{Q}^{[\boldsymbol{\Lambda}]}_{2,1}+v_{2,2}(\mathbf{W}^{(2)}_2)^\mathsf{T}\mathbf{Q}^{[\boldsymbol{\Lambda}]}_{2,2}+v_{2,4}(\mathbf{W}^{(2)}_4)^\mathsf{T}\mathbf{Q}^{[\boldsymbol{\Lambda}]}_{2,4}\\ 
    =&\frac{v_{2,1}u_{1,1}}{\alpha_2-f_1}\mathbf{W}_{1,(1)}^\mathsf{T}\boldsymbol{\lambda}_{1,(1)}
    +\frac{v_{2,1}u_{1,2}}{\alpha_2-f_2}\mathbf{W}_{1,(2)}^\mathsf{T}\boldsymbol{\lambda}_{1,(2)}
    +\frac{v_{2,2}u_{2,1}}{\alpha_2-f_1}\mathbf{W}_{2,(1)}^\mathsf{T}\boldsymbol{\lambda}_{2,(1)}\notag\\
    &+\frac{v_{2,2}u_{2,2}}{\alpha_2-f_2}\mathbf{W}_{2,(2)}^\mathsf{T}\boldsymbol{\lambda}_{2,(2)}
    +\frac{v_{2,4}u_{4,1}}{\alpha_2-f_1}\mathbf{W}_{4,(1)}^\mathsf{T}\boldsymbol{\lambda}_{4,(1)}
    +\frac{v_{2,4}u_{4,2}}{\alpha_2-f_2}\mathbf{W}_{4,(2)}^\mathsf{T}\boldsymbol{\lambda}_{4,(2)}\nonumber\\
    &+v_{2,1}I_1+v_{2,2}I_2+v_{2,2}\alpha_2 I'_2+v_{2,4}I_4+v_{2,4}\alpha_2 I'_4,
\end{align}
where 
\begin{align}
    I_2=&\mathbf{W}_{2,(1)}^\mathsf{T}\mathbf{Z}_{2,(1)}+\mathbf{W}_{2,(2)}^\mathsf{T}\mathbf{Z}_{2,(2)}\\
    I'_2=&\mathbf{W}_{2,(1)}^\mathsf{T}\mathbf{Z}'_{2,(1)}+\mathbf{W}_{2,(2)}^\mathsf{T}\mathbf{Z}'_{2,(2)}\\
    I_4=&\mathbf{W}_{4,(1)}^\mathsf{T}\mathbf{Z}_{4,(1)}+\mathbf{W}_{4,(2)}^\mathsf{T}\mathbf{Z}_{4,(2)}\\
    I'_4=&\mathbf{W}_{4,(1)}^\mathsf{T}\mathbf{Z}'_{4,(1)}+\mathbf{W}_{4,(2)}^\mathsf{T}\mathbf{Z}'_{4,(2)}.
\end{align}
The answer returned by Server $n, n\in\{3,5,6\}$ is constructed as follows.
\begin{align}
    A^{[\boldsymbol{\Lambda}]}_n
    =&v_{n,2}(\mathbf{W}^{(n)}_2)^\mathsf{T}\mathbf{Q}^{[\boldsymbol{\Lambda}]}_{n,2}+v_{n,4}(\mathbf{W}^{(n)}_4)^\mathsf{T}\mathbf{Q}^{[\boldsymbol{\Lambda}]}_{n,4}\\ 
    =&\frac{v_{n,2}u_{2,1}}{\alpha_n-f_1}\mathbf{W}_{2,(1)}^\mathsf{T}\boldsymbol{\lambda}_{2,(1)}+\frac{v_{n,2}u_{2,2}}{\alpha_n-f_2}\mathbf{W}_{2,(2)}^\mathsf{T}\boldsymbol{\lambda}_{2,(2)}
    +\frac{v_{n,4}u_{4,1}}{\alpha_n-f_1}\mathbf{W}_{4,(1)}^\mathsf{T}\boldsymbol{\lambda}_{4,(1)}\nonumber\\
    &+\frac{v_{n,4}u_{4,2}}{\alpha_n-f_2}\mathbf{W}_{4,(2)}^\mathsf{T}\boldsymbol{\lambda}_{4,(2)}+v_{n,2}I_2+v_{n,2}\alpha_n I'_2+v_{n,4}I_4+v_{n,4}\alpha_n I'_4.
\end{align}
The answer returned by Server $4$ is constructed as follows.
\begin{align}
    A^{[\boldsymbol{\Lambda}]}_4
    =&v_{4,1}(\mathbf{W}^{(4)}_1)^\mathsf{T}\mathbf{Q}^{[\boldsymbol{\Lambda}]}_{4,1}+v_{4,2}(\mathbf{W}^{(4)}_2)^\mathsf{T}\mathbf{Q}^{[\boldsymbol{\Lambda}]}_{4,2}+v_{4,3}(\mathbf{W}^{(4)}_3)^\mathsf{T}\mathbf{Q}^{[\boldsymbol{\Lambda}]}_{4,3}\\ 
    =&\frac{v_{4,1}u_{1,1}}{\alpha_4-f_1}\mathbf{W}_{1,(1)}^\mathsf{T}\boldsymbol{\lambda}_{1,(1)}
    +\frac{v_{4,1}u_{1,2}}{\alpha_4-f_2}\mathbf{W}_{1,(2)}^\mathsf{T}\boldsymbol{\lambda}_{1,(2)}
    +\frac{v_{4,2}u_{2,1}}{\alpha_4-f_1}\mathbf{W}_{2,(1)}^\mathsf{T}\boldsymbol{\lambda}_{2,(1)}\nonumber\\
    &+\frac{v_{4,2}u_{2,2}}{\alpha_4-f_2}\mathbf{W}_{2,(2)}^\mathsf{T}\boldsymbol{\lambda}_{2,(2)}
    +\frac{v_{4,3}u_{3,1}}{\alpha_4-f_1}\mathbf{W}_{3,(1)}^\mathsf{T}\boldsymbol{\lambda}_{3,(1)}
    +\frac{v_{4,3}u_{3,2}}{\alpha_4-f_2}\mathbf{W}_{3,(2)}^\mathsf{T}\boldsymbol{\lambda}_{3,(2)}\nonumber\\
    &+v_{4,1}I_1+v_{4,2}I_2+v_{4,2}\alpha_4 I'_2+v_{4,3}I_3.
\end{align}
The answer returned by Server $7$ is constructed as follows.
\begin{align}
    A^{[\boldsymbol{\Lambda}]}_7
    =&v_{7,3}(\mathbf{W}^{(7)}_3)^\mathsf{T}\mathbf{Q}^{[\boldsymbol{\Lambda}]}_{7,3}\\ 
    =&\frac{v_{7,3}u_{3,1}}{\alpha_7-f_1}\mathbf{W}_{3,(1)}^\mathsf{T}\boldsymbol{\lambda}_{3,(1)}+\frac{v_{7,3}u_{3,2}}{\alpha_7-f_2}\mathbf{W}_{3,(2)}^\mathsf{T}\boldsymbol{\lambda}_{3,(2)}+v_{7,3}I_3.
\end{align}
Upon the receiving the answers, to recover the desired linear combination, the user firstly evaluates two quantities $V_1$ and $V_2$ as follows.
\begin{align}
V_1=&\sum_{n\in[7]}A^{[\boldsymbol{\Lambda}]}_n\\
=&\sum_{m\in[4]}\Big(\Big(\sum_{n\in\mathcal{R}_m}\frac{v_{n,m}u_{m,1}}{\alpha_n-f_1}\Big)\mathbf{W}_{m,(1)}^\mathsf{T}\boldsymbol{\lambda}_{m,(1)}+\Big(\sum_{n\in\mathcal{R}_m}\frac{v_{n,m}u_{m,2}}{\alpha_n-f_2}\Big)\mathbf{W}_{m,(2)}^\mathsf{T}\boldsymbol{\lambda}_{m,(2)}\Big)\nonumber\\
    &+\Big(\sum_{n\in\mathcal{R}_1}v_{n,1}\Big)I_1
    +\Big(\sum_{n\in\mathcal{R}_2}v_{n,2}\Big)I_2
    +\Big(\sum_{n\in\mathcal{R}_2}v_{n,2}\alpha_n\Big)I'_2
    +\Big(\sum_{n\in\mathcal{R}_3}v_{n,3}\Big)I_3
    \nonumber\\
    &+\Big(\sum_{n\in\mathcal{R}_4}v_{n,4}\Big) I_4+\Big(\sum_{n\in\mathcal{R}_4}v_{n,4}\alpha_n\Big) I'_4,\\
V_2=&\sum_{n\in[7]}\alpha_n A^{[\boldsymbol{\Lambda}]}_n\\
=&\sum_{m\in[4]}\Big(\Big(\sum_{n\in\mathcal{R}_m}\frac{v_{n,m}u_{m,1}\alpha_n}{\alpha_n-f_1}\Big)\mathbf{W}_{m,(1)}^\mathsf{T}\boldsymbol{\lambda}_{m,(1)}+\Big(\sum_{n\in\mathcal{R}_m}\frac{v_{n,m}u_{m,2}\alpha_n}{\alpha_n-f_2}\Big)\mathbf{W}_{m,(2)}^\mathsf{T}\boldsymbol{\lambda}_{m,(2)}\Big)\nonumber\\
    &+\Big(\sum_{n\in\mathcal{R}_1}v_{n,1}\alpha_n\Big)I_1
    +\Big(\sum_{n\in\mathcal{R}_2}v_{n,2}\alpha_n\Big)I_2
    +\Big(\sum_{n\in\mathcal{R}_2}v_{n,2}\alpha^2_n\Big)I'_2
    +\Big(\sum_{n\in\mathcal{R}_3}v_{n,3}\alpha_n\Big)I_3
    \nonumber\\
    &+\Big(\sum_{n\in\mathcal{R}_4}v_{n,4}\alpha_n\Big) I_4+\Big(\sum_{n\in\mathcal{R}_4}v_{n,4}\alpha^2_n\Big) I'_4.
\end{align}
It is easy to verify that for all $m\in[4],\ell\in\{1,2\}$,
\begin{align}
\sum_{n\in\mathcal{R}_m} \frac{v_{n,m}u_{m,l}}{\alpha_n-f_l}=-1,\
\sum_{n\in\mathcal{R}_m} \frac{v_{n,m}u_{m,l}\alpha_n}{\alpha_n-f_l}=-f_l,\label{eq:ex1cvde}
\end{align}
and
\begin{subequations}\label{eq:ex1grs}
    \begin{align}
    &\sum_{n\in\mathcal{R}_1}v_{n,1}=0,\
    \sum_{n\in\mathcal{R}_1}v_{n,1}\alpha_n=0,\label{eq:ex1grs1}\\
    &\sum_{n\in\mathcal{R}_2}v_{n,2}=0,\
    \sum_{n\in\mathcal{R}_2}v_{n,2}\alpha_n=0,\
    \sum_{n\in\mathcal{R}_2}v_{n,2}\alpha_n^2=0,\label{eq:ex1grs2}\\
    &\sum_{n\in\mathcal{R}_3}v_{n,3}=0,\
    \sum_{n\in\mathcal{R}_3}v_{n,3}\alpha_n=0,\label{eq:ex1grs3}\\
    &\sum_{n\in\mathcal{R}_4}v_{n,4}=0,\
    \sum_{n\in\mathcal{R}_4}v_{n,4}\alpha_n=0,\
    \sum_{n\in\mathcal{R}_4}v_{n,4}\alpha_n^2=0,\label{eq:ex1grs4}
    \end{align}
\end{subequations}
where \eqref{eq:ex1cvde} is essentially due to a Vandermonde decomposition of Cauchy matrices (Lemma \ref{lemma:vdc} in Appendix), and \eqref{eq:ex1grs} follows from a dual GRS codes structure (Lemma \ref{lemma:grs} in Appendix). Therefore,
\begin{align}
V_1=&-\Big(\sum_{m\in[4]}\mathbf{W}_{m,(1)}^\mathsf{T}\boldsymbol{\lambda}_{m,(1)}\Big)-\Big(\sum_{m\in[4]}\mathbf{W}_{m,(2)}^\mathsf{T}\boldsymbol{\lambda}_{m,(2)}\Big),\\
V_2=&-f_1\Big(\sum_{m\in[4]}\mathbf{W}_{m,(1)}^\mathsf{T}\boldsymbol{\lambda}_{m,(1)}\Big)-f_2\Big(\sum_{m\in[4]}\mathbf{W}_{m,(2)}^\mathsf{T}\boldsymbol{\lambda}_{m,(2)}\Big).
\end{align}
Now the user is ready to recover the two symbols of the desired linear combination as follows.
\begin{align}
    \begin{bmatrix}
        \sum_{m\in[4]}\mathbf{W}_{m,(1)}^\mathsf{T}\boldsymbol{\lambda}_{m,(1)}\\
        \sum_{m\in[4]}\mathbf{W}_{m,(2)}^\mathsf{T}\boldsymbol{\lambda}_{m,(2)}
    \end{bmatrix}=-
    \begin{bmatrix}
        1 & 1\\
        f_1 & f_2
    \end{bmatrix}^{-1}
    \begin{bmatrix}
        V_1\\
        V_2
    \end{bmatrix}.
    \end{align}
Note that the user recovers two desired symbols from a total of $7$ downloaded symbols, the rate achieved is thus $2/7$, which matches Lemma \ref{lemma:asymm}.
\subsubsection{Motivating Example 2}\label{sec:asymmex2}
To make the general proof more accessible, let us consider another motivating example where we have $N=9$ servers and $K$ messages, $X_1=T_1=1, X_2=T_2=2$. The storage pattern for this example is defined as follows.
\begin{subequations}
    \begin{align}
        \mathcal{R}_1&=\{1,2,3,7\},\\
        \mathcal{R}_2&=\{3,4,5,6,8,9\}.
    \end{align}
\end{subequations}
The $K$ messages are partitioned into $M=2$ disjoint sets $\mathcal{W}_1$ and $\mathcal{W}_2$, where 
\begin{subequations}
    \begin{align}
        \mathcal{W}_1=\{W_{1,1},W_{1,2},\ldots,W_{1,K_1}\},\\
        \mathcal{W}_2=\{W_{2,1},W_{2,2},\ldots,W_{2,K_2}\}.
    \end{align}
\end{subequations}
Each message is comprised of $L=2$ symbols from $\mathbb{F}_q, q\geq 11$, i.e.,
\begin{align}
    W_{m,k}=(W_{m,k}(1),W_{m,k}(2)),\ \forall m\in\{1,2\},k\in[K_m].
\end{align}
Let $\alpha_1,\alpha_2,\alpha_3,\alpha_4,\alpha_5,\alpha_9,f_1,f_2$ be $11$ distinct constants from the finite field $\mathbb{F}_q$. For all $m,\ell\in\{1,2\},n\in\mathcal{R}_m$, let us define the following constants.
\begin{align}
    v_{n,m}&=\prod_{n'\in\mathcal{R}_m\setminus\{n\}}{(\alpha_n-\alpha_{n'})^{-1}},\\
    u_{m,l}&=\prod_{n\in\mathcal{R}_m}{(f_l-\alpha_n)}.
\end{align}
Let $\mathbf{\dot{Z}}_{1,(1)}\in\mathbb{F}_q^{K_1},\mathbf{\dot{Z}}_{1,(2)}\in\mathbb{F}_q^{K_1},\mathbf{\dot{Z}}_{2,(1)}\in\mathbb{F}_q^{K_2},\mathbf{\dot{Z}}_{2,(2)}\in\mathbb{F}_q^{K_2},\mathbf{\ddot{Z}}_{2,(1)}\in\mathbb{F}_q^{K_2},\mathbf{\ddot{Z}}_{2,(2)}\in\mathbb{F}_q^{K_2}$ be i.i.d. uniform column row vectors that are independent of the messages. The storage at the $9$ servers, i.e. $\mathcal{S}_1,\mathcal{S}_2,\ldots,\mathcal{S}_9$  are constructed as follows.
\begin{subequations}
    \begin{align}
        \mathcal{S}_1&=\{\mathbf{W}^{(1)}_1\},\ &
        \mathcal{S}_2&=\{\mathbf{W}^{(2)}_1\},\ &
        \mathcal{S}_3&=\{\mathbf{W}^{(3)}_1,\mathbf{W}^{(3)}_2\},\\
        \mathcal{S}_4&=\{\mathbf{W}^{(4)}_2\},\ &
        \mathcal{S}_5&=\{\mathbf{W}^{(5)}_2\},\ &
        \mathcal{S}_6&=\{\mathbf{W}^{(6)}_2\},\\
        \mathcal{S}_7&=\{\mathbf{W}^{(7)}_1\},\ &
        \mathcal{S}_8&=\{\mathbf{W}^{(8)}_2\},\ &
        \mathcal{S}_9&=\{\mathbf{W}^{(9)}_2\},
    \end{align}
\end{subequations}
where for all $n\in\{1,2,3,7\}$,
\begin{align}
    \mathbf{W}^{(n)}_1=
    \begin{bmatrix}
        \frac{1}{\alpha_n-f_1}\mathbf{W}_{1,(1)}+\mathbf{\dot{Z}}_{1,(1)}\\
        \frac{1}{\alpha_n-f_2}\mathbf{W}_{1,(2)}+\mathbf{\dot{Z}}_{1,(2)}
    \end{bmatrix},
\end{align}
and for all $n\in\{3,4,5,6,8,9\}$,
\begin{align}
    \mathbf{W}^{(n)}_2=
    \begin{bmatrix}
        \frac{1}{\alpha_n-f_1}\mathbf{W}_{2,(1)}+\mathbf{\dot{Z}}_{2,(1)}+\alpha_n\mathbf{\ddot{Z}}_{2,(1)}\\
        \frac{1}{\alpha_n-f_2}\mathbf{W}_{2,(2)}+\mathbf{\dot{Z}}_{2,(2)}+\alpha_n\mathbf{\ddot{Z}}_{2,(2)}
    \end{bmatrix},
\end{align}
and for all $m\in[2],\ell\in[2]$,
\begin{align}
    \mathbf{W}_{m,(\ell)}=\left[W_{m,1}(\ell),\cdots,W_{m,k_m}(\ell)\right]^\mathsf{T}.
\end{align}
It is easy to see that the $X_1=1$ security constraint is guaranteed by the random noise vectors $\mathbf{\dot{Z}}_{1,(1)},\mathbf{\dot{Z}}_{1,(2)}$ protecting the messages $\mathcal{W}_1$, the $X_2=2$ security constraint is guaranteed by the MDS coded random noise vectors $\mathbf{\dot{Z}}_{2,(1)},\mathbf{\dot{Z}}_{2,(2)},\mathbf{\ddot{Z}}_{2,(1)},\mathbf{\ddot{Z}}_{2,(2)}$ protecting the messages $\mathcal{W}_2$. To retrieve the desired linear combination of messages, the user generates i.i.d. uniform column vectors $\mathbf{Z}'_{1,(1)}\in\mathbb{F}_q^{K_1},\mathbf{Z}'_{1,(2)}\in\mathbb{F}_q^{K_1},\mathbf{Z}'_{2,(1)}\in\mathbb{F}_q^{K_2},\mathbf{Z}'_{2,(2)}\in\mathbb{F}_q^{K_2},\mathbf{Z}''_{2,(1)}\in\mathbb{F}_q^{K_2},\mathbf{Z}''_{2,(2)}\in\mathbb{F}_q^{K_2}$ privately that are independent of $\boldsymbol{\Lambda}$.
The queries $Q_{[N]}^{[\boldsymbol{\Lambda}]}$ are constructed as follows.
\begin{subequations}
    \begin{align}
        Q_1^{[\boldsymbol{\Lambda}]}&=
        \{\mathbf{Q}^{[\boldsymbol{\Lambda}]}_{1,1}\},&\
        Q_2^{[\boldsymbol{\Lambda}]}&=
        \{\mathbf{Q}^{[\boldsymbol{\Lambda}]}_{2,1}\},&\
        Q_3^{[\boldsymbol{\Lambda}]}&=
        \{\mathbf{Q}^{[\boldsymbol{\Lambda}]}_{3,1},\mathbf{Q}^{[\boldsymbol{\Lambda}]}_{3,2}\},\\
        Q_4^{[\boldsymbol{\Lambda}]}&=
        \{\mathbf{Q}^{[\boldsymbol{\Lambda}]}_{4,2}\},&\
        Q_5^{[\boldsymbol{\Lambda}]}&=
        \{\mathbf{Q}^{[\boldsymbol{\Lambda}]}_{5,2}\},&\
        Q_6^{[\boldsymbol{\Lambda}]}&=
        \{\mathbf{Q}^{[\boldsymbol{\Lambda}]}_{6,2}\},\\
        Q_7^{[\boldsymbol{\Lambda}]}&=
        \{\mathbf{Q}^{[\boldsymbol{\Lambda}]}_{7,1}\},&\
        Q_8^{[\boldsymbol{\Lambda}]}&=
        \{\mathbf{Q}^{[\boldsymbol{\Lambda}]}_{8,2}\},&\
        Q_9^{[\boldsymbol{\Lambda}]}&=
        \{\mathbf{Q}^{[\boldsymbol{\Lambda}]}_{9,2}\}.
    \end{align}
\end{subequations}
where for all $n\in\{1,2,3,7\}$,
\begin{align}
    \mathbf{Q}^{[\boldsymbol{\Lambda}]}_{n,1}=
    \begin{bmatrix}
        u_{1,1}\boldsymbol{\lambda}_{1,(1)}+(\alpha_n-f_1)\mathbf{Z}'_{1,(1)}\\
        u_{1,2}\boldsymbol{\lambda}_{1,(2)}+(\alpha_n-f_2)\mathbf{Z}'_{1,(2)}
    \end{bmatrix},
\end{align}
and for all $n\in\{3,4,5,6,8,9\}$,
\begin{align}
    \mathbf{Q}^{[\boldsymbol{\Lambda}]}_{n,2}=
    \begin{bmatrix}
        u_{2,1}\boldsymbol{\lambda}_{2,(1)}+(\alpha_n-f_1)(\mathbf{Z}'_{2,(1)}+\alpha_n\mathbf{Z}''_{2,(1)})\\
        u_{2,2}\boldsymbol{\lambda}_{2,(2)}+(\alpha_n-f_2)(\mathbf{Z}'_{2,(2)}+\alpha_n\mathbf{Z}''_{2,(2)})
    \end{bmatrix},
\end{align}
and for all $m\in[2],\ell\in[2]$,
\begin{align}
    \boldsymbol{\lambda}_{m,(\ell)}=[\lambda_{m,1,(\ell)},\cdots,\lambda_{m,K_m,(\ell)}]^\mathsf{T}.
\end{align}
Similarly, the $T_1=1$ and $T_2=2$ privacy constraints are guaranteed by the random noise vectors protecting the coefficients of the desired linear combination. Note that the desired linear combination is equivalently written in terms of $(\mathbf{W}_{m,(\ell)})_{m\in[2],\ell\in[2]}$ and $(\boldsymbol{\lambda}_{m,(\ell)})_{m\in[2],\ell\in[2]}$ as follows.
\begin{align}
    \lambda_{\boldsymbol{\Lambda}}(\mathcal{W})\hspace{-0.1cm}&=\hspace{-0.1cm}
    \begin{bmatrix}
        \sum_{k\in[K_1]} \lambda_{1,k,(1)}W_{1,k}(1)+\sum_{k\in[K_2]} \lambda_{2,k,(1)}W_{2,k}(1)\\
        \sum_{k\in[K_1]} \lambda_{1,k,(2)}W_{1,k}(2)+\sum_{k\in[K_2]} \lambda_{2,k,(2)}W_{2,k}(2)
    \end{bmatrix}\\&=
    \begin{bmatrix}
        \mathbf{W}_{1,(1)}^\mathsf{T}\boldsymbol{\lambda}_{1,(1)}+\mathbf{W}_{2,(1)}^\mathsf{T}\boldsymbol{\lambda}_{2,(1)}\\
        \mathbf{W}_{1,(2)}^\mathsf{T}\boldsymbol{\lambda}_{1,(2)}+\mathbf{W}_{2,(2)}^\mathsf{T}\boldsymbol{\lambda}_{2,(2)}
    \end{bmatrix}.
\end{align}

The query $Q_n^{[\boldsymbol{\Lambda}]}$ is sent to Server $n$ for all $n\in[9]$. Upon receiving the query $Q_n^{[\boldsymbol{\Lambda}]}$, Server $n, n\in \{1,2,7\}$ responds the user with the following answer.
\begin{align}
    A^{[\boldsymbol{\Lambda}]}_n
    =&v_{n,1}(\mathbf{W}^{(n)}_1)^\mathsf{T}\mathbf{Q}^{[\boldsymbol{\Lambda}]}_{n,1}\\
    =&\frac{v_{n,1}u_{1,1}}{\alpha_n-f_1}\mathbf{W}_{1,(1)}^\mathsf{T}\boldsymbol{\lambda}_{1,(1)}+\frac{v_{n,1}u_{1,2}}{\alpha_n-f_2}\mathbf{W}_{1,(2)}^\mathsf{T}\boldsymbol{\lambda}_{1,(2)}+v_{n,1}I_{1,1}+v_{n,1}\alpha_n I_{1,2},
\end{align}
where
\begin{align}
    I_{1,1}=&\mathbf{W}_{1,(1)}^\mathsf{T}\mathbf{Z}'_{1,(1)}+u_{1,1}\mathbf{\dot{Z}}_{1,(1)}^\mathsf{T}\boldsymbol{\lambda}_{1,(1)}-f_1\mathbf{\dot{Z}}_{1,(1)}^\mathsf{T}\mathbf{Z}'_{1,(1)}\nonumber\\
    &+\mathbf{W}_{1,(2)}^\mathsf{T}\mathbf{Z}'_{1,(2)}+u_{1,2}\mathbf{\dot{Z}}_{1,(2)}^\mathsf{T}\boldsymbol{\lambda}_{1,(2)}-f_2\mathbf{\dot{Z}}_{1,(2)}^\mathsf{T}\mathbf{Z}'_{1,(2)},\\
    I_{1,2}=&\mathbf{\dot{Z}}_{1,(1)}^\mathsf{T}\mathbf{Z}'_{1,(1)}+\mathbf{\dot{Z}}_{1,(2)}^\mathsf{T}\mathbf{Z}'_{1,(2)}.
\end{align}
The answer returned by Server $3$ is constructed as follows.
\begin{align}
    A^{[\boldsymbol{\Lambda}]}_3
    =&v_{3,1}(\mathbf{W}^{(3)}_1)^\mathsf{T}\mathbf{Q}^{[\boldsymbol{\Lambda}]}_{3,1}+v_{3,2}(\mathbf{W}^{(3)}_2)^\mathsf{T}\mathbf{Q}^{[\boldsymbol{\Lambda}]}_{3,2}\\ 
    =&\frac{v_{3,1}u_{1,1}}{\alpha_3-f_1}\mathbf{W}_{1,(1)}^\mathsf{T}\boldsymbol{\lambda}_{1,(1)}
    +\frac{v_{3,1}u_{1,2}}{\alpha_3-f_2}\mathbf{W}_{1,(2)}^\mathsf{T}\boldsymbol{\lambda}_{1,(2)}
    +\frac{v_{3,2}u_{2,1}}{\alpha_3-f_1}\mathbf{W}_{2,(1)}^\mathsf{T}\boldsymbol{\lambda}_{2,(1)}\nonumber\\
    &+\frac{v_{3,2}u_{2,2}}{\alpha_3-f_2}\mathbf{W}_{2,(2)}^\mathsf{T}\boldsymbol{\lambda}_{2,(2)}
    +v_{3,1}I_{1,1}+v_{3,1}\alpha_3 I_{1,2}+v_{3,2}I_{2,1}+v_{3,2}\alpha_3 I_{2,2}\nonumber\\
    &+v_{3,2}\alpha^2_3 I_{2,3}+v_{3,2}\alpha^3_3 I_{2,4},
\end{align}
where 
\begin{align}
    I_{2,1}=&\mathbf{W}_{2,(1)}^\mathsf{T}\mathbf{Z}'_{2,(1)}+u_{2,1}\mathbf{\dot{Z}}_{2,(1)}^\mathsf{T}\boldsymbol{\lambda}_{2,(1)}+u_{2,1}\mathbf{\ddot{Z}}_{2,(1)}^\mathsf{T}\boldsymbol{\lambda}_{2,(1)}-f_1\mathbf{\dot{Z}}_{2,(1)}^\mathsf{T}\mathbf{Z}'_{2,(1)}\nonumber\\
    &+\mathbf{W}_{2,(2)}^\mathsf{T}\mathbf{Z}'_{2,(2)}+u_{2,2}\mathbf{\dot{Z}}_{2,(2)}^\mathsf{T}\boldsymbol{\lambda}_{2,(2)}+u_{2,2}\mathbf{\ddot{Z}}_{2,(2)}^\mathsf{T}\boldsymbol{\lambda}_{2,(2)}-f_2\mathbf{\dot{Z}}_{2,(2)}^\mathsf{T}\mathbf{Z}'_{2,(2)},\\
    I_{2,2}=&\mathbf{W}_{2,(1)}^\mathsf{T}\mathbf{Z}''_{2,(1)}+\alpha_n\mathbf{\dot{Z}}_{2,(1)}^\mathsf{T}\mathbf{Z}'_{2,(1)}-f_1\mathbf{\dot{Z}}_{2,(1)}^\mathsf{T}\mathbf{Z}''_{2,(1)}-f_1\mathbf{\ddot{Z}}_{2,(1)}^\mathsf{T}\mathbf{Z}'_{2,(1)}\nonumber\\
    &+\mathbf{W}_{2,(2)}^\mathsf{T}\mathbf{Z}''_{2,(2)}+\alpha_n\mathbf{\dot{Z}}_{2,(2)}^\mathsf{T}\mathbf{Z}'_{2,(2)}-f_2\mathbf{\dot{Z}}_{2,(2)}^\mathsf{T}\mathbf{Z}''_{2,(2)}-f_2\mathbf{\ddot{Z}}_{2,(2)}^\mathsf{T}\mathbf{Z}'_{2,(2)},\\
    I_{2,3}=&\mathbf{\dot{Z}}_{2,(1)}^\mathsf{T}\mathbf{Z}''_{2,(1)}+\mathbf{\ddot{Z}}_{2,(1)}^\mathsf{T}\mathbf{Z}'_{2,(1)}-f_1\mathbf{\ddot{Z}}_{2,(1)}^\mathsf{T}\mathbf{Z}''_{2,(1)}\nonumber\\
    &+\mathbf{\dot{Z}}_{2,(2)}^\mathsf{T}\mathbf{Z}''_{2,(2)}+\mathbf{\ddot{Z}}_{2,(2)}^\mathsf{T}\mathbf{Z}'_{2,(2)}-f_2\mathbf{\ddot{Z}}_{2,(2)}^\mathsf{T}\mathbf{Z}''_{2,(2)},\\
    I_{2,4}=&\mathbf{\ddot{Z}}_{2,(1)}^\mathsf{T}\mathbf{Z}''_{2,(1)}+\mathbf{\ddot{Z}}_{2,(2)}^\mathsf{T}\mathbf{Z}''_{2,(2)}.
\end{align}
The answer returned by Server $n, n\in\{4,5,6,8,9\}$ is constructed as follows.
\begin{align}
    A^{[\boldsymbol{\Lambda}]}_n
    =&v_{n,2}(\mathbf{W}^{(n)}_2)^\mathsf{T}\mathbf{Q}^{[\boldsymbol{\Lambda}]}_{n,2}\\ 
    =&\frac{v_{n,2}u_{2,1}}{\alpha_n-f_1}\mathbf{W}_{2,(1)}^\mathsf{T}\boldsymbol{\lambda}_{2,(1)}+\frac{v_{n,2}u_{2,2}}{\alpha_n-f_2}\mathbf{W}_{2,(2)}^\mathsf{T}\boldsymbol{\lambda}_{2,(2)}+v_{n,2}I_{2,1}+v_{n,2}\alpha_n I_{2,2}\nonumber\\
    &+v_{n,2}\alpha^2_n I_{2,3}+v_{n,2}\alpha^3_n I_{2,4}.
\end{align}
Upon the receiving the answers, to recover the desired linear combination, the user firstly evaluates two quantities $V_1$ and $V_2$ as follows.
\begin{align}
V_1=&\sum_{n\in[9]}A^{[\boldsymbol{\Lambda}]}_n\\
=&\Big(\sum_{n\in\mathcal{R}_1}\frac{v_{n,1}u_{1,1}}{\alpha_n-f_1}\Big)\mathbf{W}_{1,(1)}^\mathsf{T}\boldsymbol{\lambda}_{1,(1)}+\Big(\sum_{n\in\mathcal{R}_1}\frac{v_{n,1}u_{1,2}}{\alpha_n-f_2}\Big)\mathbf{W}_{1,(2)}^\mathsf{T}\boldsymbol{\lambda}_{1,(2)}\nonumber\\
&+\Big(\sum_{n\in\mathcal{R}_2}\frac{v_{n,2}u_{2,1}}{\alpha_n-f_1}\Big)\mathbf{W}_{2,(1)}^\mathsf{T}\boldsymbol{\lambda}_{2,(1)}+\Big(\sum_{n\in\mathcal{R}_2}\frac{v_{n,2}u_{2,2}}{\alpha_n-f_2}\Big)\mathbf{W}_{2,(2)}^\mathsf{T}\boldsymbol{\lambda}_{2,(2)}\nonumber\\
    &+\Big(\sum_{n\in\mathcal{R}_1}v_{n,1}\Big)I_{1,1}+\Big(\sum_{n\in\mathcal{R}_1}v_{n,1}\alpha_n\Big) I_{1,2}+\Big(\sum_{n\in\mathcal{R}_2}v_{n,2}\Big)I_{2,1}+\Big(\sum_{n\in\mathcal{R}_2}v_{n,2}\alpha_n\Big) I_{2,2}\nonumber\\
    &+\Big(\sum_{n\in\mathcal{R}_2}v_{n,2}\alpha^2_n\Big) I_{2,3}+\Big(\sum_{n\in\mathcal{R}_2}v_{n,2}\alpha^3_n\Big) I_{2,4},\\
V_2=&\sum_{n\in[9]}\alpha_n A^{[\boldsymbol{\Lambda}]}_n\\
=&\Big(\sum_{n\in\mathcal{R}_1}\frac{v_{n,1}u_{1,1}\alpha_n}{\alpha_n-f_1}\Big)\mathbf{W}_{1,(1)}^\mathsf{T}\boldsymbol{\lambda}_{1,(1)}+\Big(\sum_{n\in\mathcal{R}_1}\frac{v_{n,1}u_{1,2}\alpha_n}{\alpha_n-f_2}\Big)\mathbf{W}_{1,(2)}^\mathsf{T}\boldsymbol{\lambda}_{1,(2)}\nonumber\\
&+\Big(\sum_{n\in\mathcal{R}_2}\frac{v_{n,2}u_{2,1}\alpha_n}{\alpha_n-f_1}\Big)\mathbf{W}_{2,(1)}^\mathsf{T}\boldsymbol{\lambda}_{2,(1)}+\Big(\sum_{n\in\mathcal{R}_2}\frac{v_{n,2}u_{2,2}\alpha_n}{\alpha_n-f_2}\Big)\mathbf{W}_{2,(2)}^\mathsf{T}\boldsymbol{\lambda}_{2,(2)}\nonumber\\
    &+\Big(\sum_{n\in\mathcal{R}_1}v_{n,1}\alpha_n\Big)I_{1,1}+\Big(\sum_{n\in\mathcal{R}_1}v_{n,1}\alpha^2_n\Big) I_{1,2}+\Big(\sum_{n\in\mathcal{R}_2}v_{n,2}\alpha_n\Big)I_{2,1}+\Big(\sum_{n\in\mathcal{R}_2}v_{n,2}\alpha^2_n\Big) I_{2,2}\nonumber\\
    &+\Big(\sum_{n\in\mathcal{R}_2}v_{n,2}\alpha^3_n\Big) I_{2,3}+\Big(\sum_{n\in\mathcal{R}_2}v_{n,2}\alpha^4_n\Big) I_{2,4}
\end{align}
Again, due to Lemma \ref{lemma:grs} and Lemma \ref{lemma:vdc}, for all $m\in\{1,2\},\ell\in\{1,2\}$, we have the following identities.
\begin{align}
\sum_{n\in\mathcal{R}_m} \frac{v_{n,m}u_{m,l}}{\alpha_n-f_l}=-1,\
\sum_{n\in\mathcal{R}_m} \frac{v_{n,m}u_{m,l}\alpha_n}{\alpha_n-f_l}=-f_l,\label{eq:excvde}
\end{align}
and
\begin{subequations}\label{eq:exgrs1}
    \begin{align}
    &\sum_{n\in\mathcal{R}_1}v_{n,1}=0,\
    \sum_{n\in\mathcal{R}_1}v_{n,1}\alpha_n=0,\
    \sum_{n\in\mathcal{R}_1}v_{n,1}\alpha_n^2=0,\\
    &\sum_{n\in\mathcal{R}_2}v_{n,2}=0,\
    \sum_{n\in\mathcal{R}_2}v_{n,2}\alpha_n=0,\
    \sum_{n\in\mathcal{R}_2}v_{n,2}\alpha_n^2=0,\\
    &\sum_{n\in\mathcal{R}_2}v_{n,2}\alpha_n^3=0,\
    \sum_{n\in\mathcal{R}_2}v_{n,2}\alpha_n^4=0.
    \end{align}
\end{subequations}
Therefore,
\begin{align}
V_1=&-\left(\mathbf{W}_{1,(1)}^\mathsf{T}\boldsymbol{\lambda}_{1,(1)}+\mathbf{W}_{2,(1)}^\mathsf{T}\boldsymbol{\lambda}_{2,(1)}\right)-\left(\mathbf{W}_{1,(2)}^\mathsf{T}\boldsymbol{\lambda}_{1,(2)}+\mathbf{W}_{2,(2)}^\mathsf{T}\boldsymbol{\lambda}_{2,(2)}\right),\\
V_2=&-f_1\left(\mathbf{W}_{1,(1)}^\mathsf{T}\boldsymbol{\lambda}_{1,(1)}+\mathbf{W}_{2,(1)}^\mathsf{T}\boldsymbol{\lambda}_{2,(1)}\right)
-f_2\left(\mathbf{W}_{1,(2)}^\mathsf{T}\boldsymbol{\lambda}_{1,(2)}+\mathbf{W}_{2,(2)}^\mathsf{T}\boldsymbol{\lambda}_{2,(2)}\right).
\end{align}
Now the user is ready to recover the two symbols of the desired linear combination as follows.
\begin{align}
    \begin{bmatrix}
        \mathbf{W}_{1,(1)}^\mathsf{T}\boldsymbol{\lambda}_{1,(1)}+\mathbf{W}_{2,(1)}^\mathsf{T}\boldsymbol{\lambda}_{2,(1)}\\
        \mathbf{W}_{1,(2)}^\mathsf{T}\boldsymbol{\lambda}_{1,(2)}+\mathbf{W}_{2,(2)}^\mathsf{T}\boldsymbol{\lambda}_{2,(2)}
    \end{bmatrix}=-
    \begin{bmatrix}
        1 & 1\\
        f_1 & f_2
    \end{bmatrix}^{-1}
    \begin{bmatrix}
        V_1\\
        V_2
    \end{bmatrix}.
    \end{align}
Note that the user recovers two desired symbols from a total of $9$ downloaded symbols, the rate achieved is thus $2/9$, which matches Lemma \ref{lemma:asymm}.

\subsubsection{The General Scheme}
Let us assume $\min_{m\in[M]}(\rho_m-X_m-T_m)>0$ and set $L=\min_{m\in[M]}(\rho_m-X_m-T_m)$. The scheme needs a total of $N+L$ distinct constants $\alpha_1,\ldots,\alpha_N$ and $f_1,f_2,\cdots,f_L$ from a finite field $\mathbb{F}_q$, thus we require $q\ge N+L$. 
For all $m\in[M],\ell\in[L]$, let us define
\begin{align}
    \mathbf{W}_{m,(\ell)}&=[W_{m,1}(\ell),W_{m,2}(\ell),\cdots,W_{m,K_m}(\ell)]^\mathsf{T},\\
    \boldsymbol{\lambda}_{m,(\ell)}&=[\lambda_{m,1,(\ell)},\cdots,\lambda_{m,K_m,(\ell)}]^\mathsf{T}.
\end{align}
Therefore, the desired linear combination can be equivalently written as
\begin{align}
    \lambda_{\boldsymbol{\Lambda}}(\mathcal{W})&=
    \begin{bmatrix}
        \sum_{m\in[M]}\sum_{k\in[K_m]} \lambda_{m,k,(1)}W_{m,k}(1)\\
        \sum_{m\in[M]}\sum_{k\in[K_m]} \lambda_{m,k,(2)}W_{m,k}(2)\\
        \vdots\\
        \sum_{m\in[M]}\sum_{k\in[K_m]} \lambda_{m,k,(L)}W_{m,k}(L)\\
    \end{bmatrix}\\&=
    \left[
        \sum_{m\in[M]}\mathbf{W}_{m,(1)}^\mathsf{T}\boldsymbol{\lambda}_{m,(1)},
        \sum_{m\in[M]}\mathbf{W}_{m,(2)}^\mathsf{T}\boldsymbol{\lambda}_{m,(2)},
        \cdots,
        \sum_{m\in[M]}\mathbf{W}_{m,(L)}^\mathsf{T}\boldsymbol{\lambda}_{m,(L)}
    \right]^\mathsf{T}.
\end{align}
For all $m\in[M],x\in[X_m],\ell\in[L]$, let $\mathbf{Z}_{m,x,(\ell)}$ be i.i.d. uniform column vectors from $\mathbb{F}^{K_m}_q$ that are independent of messages. 
For all $n\in[N]$, the storage at Server $n$ is $\mathcal{S}_n=\{\mathbf{W}^{(n)}_{m} | m\in\mathcal{M}_n\}$, where for all $m\in\mathcal{M}_n$,
\begin{align}
    \mathbf{W}^{(n)}_{m}=\begin{bmatrix}
        \frac{1}{\alpha_{n}-f_1}\mathbf{W}_{m,(1)}+\sum_{x\in[X_m]}\alpha_{n}^{x-1}\mathbf{Z}_{m,x,(1)}\\
        \frac{1}{\alpha_{n}-f_2}\mathbf{W}_{m,(2)}+\sum_{x\in[X_m]}\alpha_{n}^{x-1}\mathbf{Z}_{m,x,(2)}\\
        \vdots\\
        \frac{1}{\alpha_{n}-f_L}\mathbf{W}_{m,(L)}+\sum_{x\in[X_m]}\alpha_{n}^{x-1}\mathbf{Z}_{m,x,(L)}\\
    \end{bmatrix}  
    .\label{eq:svec}
\end{align}
Note that the share $\{\mathbf{W}^{(n)}_{m}\}_{n\in\mathcal{R}_m}$ is $X_m$-secure due to the MDS($\rho_m,X_m$) coded random noise vectors protecting the message vectors $\mathbf{W}_{m,(\ell)}$. 
For all $m\in[M],\ell\in[L],t\in[T_m]$, let $\mathbf{Z'}_{m,t,(\ell)}$ be i.i.d. uniform column vectors from $\mathbb{F}^{K_m}_q$ that are independent of the coefficients $\boldsymbol{\Lambda}$, generated by the user privately.  
For all $m\in[M],\ell\in[L]$, let us define the following constants,
\begin{align}
    u_{m,\ell}=\prod_{n\in\mathcal{R}_m}(f_\ell-\alpha_{n}).
\end{align}
For all $n\in[N]$, the query sent to Server $n$ is $Q^{[\boldsymbol{\Lambda}]}_n=\{\mathbf{Q}^{[\boldsymbol{\Lambda}]}_{n,m} | m\in\mathcal{M}_n\}$, where $m\in\mathcal{M}_n$,
\begin{align}
    \mathbf{Q}^{[\boldsymbol{\Lambda}]}_{n,m}=\begin{bmatrix}
    u_{m,1}\boldsymbol{\lambda}_{m,(1)}+(\alpha_{n}-f_1)\sum_{t\in[T_m]}\alpha_{n}^{t-1} \mathbf{Z'}_{m,t,(1)}\\
    u_{m,2}\boldsymbol{\lambda}_{m,(2)}+(\alpha_{n}-f_2)\sum_{t\in[T_m]}\alpha_{n}^{t-1} \mathbf{Z'}_{m,t,(2)}\\
    \vdots\\
    u_{m,L}\boldsymbol{\lambda}_{m,(L)}+(\alpha_{n}-f_L)\sum_{t\in[T_m]}\alpha_{n}^{t-1} \mathbf{Z'}_{m,t,(L)}
\end{bmatrix}.
    \label{eq:qvec}
\end{align}
Similarly, the query $\{\mathbf{Q}^{[\boldsymbol{\Lambda}]}_{n,m}\}_{n\in\mathcal{R}_m}$ for the $m^{th}$ message set is $T_m$-private since $\boldsymbol{\lambda}_{m,(\ell)}$ is protected by the MDS($\rho_m,T_m$) coded random noise vectors. 
For all $m\in[M],n\in\mathcal{R}_m$, let us define the following constants,
\begin{align}
    v_{n,m}=\prod_{n'\in\mathcal{R}_m\setminus\{n\}}(\alpha_n-\alpha_{n'})^{-1}.\label{vmn}
\end{align}
Finally, for all $n\in[N]$, the answer returned by Server $n$ is constructed as follows.
\begin{align}
    A^{[\boldsymbol{\Lambda}]}_{n}
    =&\sum_{m\in\mathcal{M}_n}v_{n,m}(\mathbf{W}^{(n)}_{m})^\mathsf{T}\mathbf{Q}^{[\boldsymbol{\Lambda}]}_{n,m}\label{ans}\\
    =&\sum_{m\in\mathcal{M}_n}\sum_{\ell\in[L]}v_{n,m}\Big(\frac{1}{\alpha_{n}-f_\ell}\mathbf{W}_{m,(\ell)}+\sum_{x\in[X_m]}\alpha_{n}^{x-1}\mathbf{Z}_{m,x,(\ell)}\Big)^\mathsf{T}\nonumber\\
    &\hspace{2.5cm}\times\Big(u_{m,\ell}\boldsymbol{\lambda}_{m,(\ell)}+(\alpha_{n}-f_\ell)\sum_{t\in[T_m]}\alpha_{n}^{t-1} \mathbf{Z'}_{m,t,(\ell)}\Big)\\ 
    =&\sum_{m\in\mathcal{M}_n}\sum_{\ell\in[L]}v_{n,m}\Big(\frac{u_{m,\ell}}{\alpha_{n}-f_\ell}\mathbf{W}_{m,(\ell)}^\mathsf{T}\boldsymbol{\lambda}_{m,(\ell)}+u_{m,\ell}\sum_{x\in[X_m]}\alpha_{n}^{x-1}\mathbf{Z}_{m,x,(\ell)}^\mathsf{T}\boldsymbol{\lambda}_{m,(\ell)}\nonumber\\
    &\hspace{3cm}+\sum_{t\in[T_m]}\alpha_{n}^{t-1}\mathbf{W}_{m,(\ell)}^\mathsf{T}\mathbf{Z'}_{m,t,(\ell)}\nonumber\\
    &\hspace{3cm}+(\alpha_{n}-f_\ell)\sum_{x\in[X_m]}\sum_{t\in[T_m]}\alpha_{n}^{x+t-2}\mathbf{Z}_{m,x,(\ell)}^\mathsf{T}\mathbf{Z'}_{m,t,(\ell)}\Big)\\ 
    =&\sum_{m\in\mathcal{M}_n}\sum_{\ell\in[L]}v_{n,m}\Big(\frac{u_{m,\ell}}{\alpha_n-f_\ell}\mathbf{W}_{m,(\ell)}^\mathsf{T}\boldsymbol{\lambda}_{m,(\ell)}+\sum_{z\in[X_m+T_m]} \alpha^{z-1}_n I_{m,z,(\ell)}\Big)\\
    =&\sum_{m\in\mathcal{M}_n}\Big(\sum_{\ell\in[L]}\frac{v_{n,m}u_{m,\ell}}{\alpha_n-f_\ell}\mathbf{W}_{m,(\ell)}^\mathsf{T}\boldsymbol{\lambda}_{m,(\ell)}+\sum_{z\in[X_m+T_m]} v_{n,m}\alpha^{z-1}_n I'_{m,z}\Big),
\end{align}
where for all $z\in[X_m+T_m]$, the interference terms $I_{m,z,(\ell)}$ are various linear combinations of inner products of $\mathbf{W}_{m,(\ell)}$, $\boldsymbol{\lambda}_{m,(\ell)}$, $(\mathbf{Z}_{m,x,(\ell)})_{x\in[X_m]}$ and $(\mathbf{Z'}_{m,t,(\ell)})_{t\in[T_m]}$, whose exact forms are irrelevant. To see the correctness of our scheme, we note that upon receiving the answers, the user can evaluate the quantities $V_1,\cdots,V_L$ defined as $V_i=\sum_{n\in[N]}\alpha^{i-1}_n A^{[\boldsymbol{\Lambda}]}_n, i\in[L].$
Note that for all $i\in[L]$, we have
\begin{align}
    V_i
    =&\sum_{n\in[N]}\alpha^{i-1}_n A^{[\boldsymbol{\Lambda}]}_n\\
    =&\sum_{m\in\mathcal{M}_n}\Big(\sum_{\ell\in[L]}\frac{v_{n,m}u_{m,\ell}\alpha^{i-1}_n}{\alpha_n-f_\ell}\mathbf{W}_{m,(\ell)}^\mathsf{T}\boldsymbol{\lambda}_{m,(\ell)}+\sum_{z\in[X_m+T_m]} v_{n,m}\alpha^{z-1}_n I'_{m,z}\Big)\\
    =&\sum_{m\in[M]}\sum_{\ell\in[L]}\mathbf{W}_{m,(\ell)}^\mathsf{T}\boldsymbol{\lambda}_{m,(\ell)}
    \Big(\sum_{n\in\mathcal{R}_m}
    \frac{v_{n,m}u_{m,\ell}\alpha_n^{i-1}}{\alpha_n-f_\ell}\Big)\nonumber\\
    &+\sum_{m\in[M]}\sum_{z\in[X_m+T_m]} \Big(\sum_{n\in\mathcal{R}_m}v_{n,m}\alpha^{i+z-2}_n\Big)I'_{m,z}.
    \label{eq:ia1}
\end{align} 
According to Lemma \ref{lemma:grs} in Appendix, namely, the dual GRS codes structure, since $0\le i+z-2\le L+X_m+T_m-2 = \rho_{\textrm{m}}-2$, for all $m\in[M],i\in[L],z\in[X_m+T_m]$, we have
\begin{align}
    \sum_{n\in\mathcal{R}_m}v_{n,m}\alpha^{i+z-2}_n=0,
\end{align}
i.e., the second term in \eqref{eq:ia1} is zero, therefore for all $i\in[L]$,
\begin{align}
    V_i=\sum_{m\in[M]}\sum_{\ell\in[L]}\mathbf{W}_{m,(\ell)}^\mathsf{T}\boldsymbol{\lambda}_{m,(\ell)}
    \Big(\sum_{n\in\mathcal{R}_m}
    \frac{v_{n,m}u_{m,\ell}\alpha_n^{i-1}}{\alpha_n-f_\ell}\Big).
\end{align}   
For all $m\in[M]$, define
\begin{align}
    &\mathbf{C}_m=\begin{bmatrix}
        \frac{1}{\alpha_{\mathcal{R}_m(1)}-f_1} & \frac{1}{\alpha_{\mathcal{R}_m(1)}-f_2} & \cdots & \frac{1}{\alpha_{\mathcal{R}_m(1)}-f_L}\\
        \frac{1}{\alpha_{\mathcal{R}_m(2)}-f_1} & \frac{1}{\alpha_{\mathcal{R}_m(2)}-f_2} & \cdots & \frac{1}{\alpha_{\mathcal{R}_m(2)}-f_L}\\
        \vdots & \vdots & \vdots & \vdots \\
        \frac{1}{\alpha_{\mathcal{R}_m(\rho_m)}-f_1} & \frac{1}{\alpha_{\mathcal{R}_m(\rho_m)}-f_2} & \cdots & \frac{1}{\alpha_{\mathcal{R}_m(\rho_m)}-f_L}\\
    \end{bmatrix},\\
    &\mathbf{V}_m=
    \begin{bmatrix}
        1 & 1 & \cdots & 1\\
        \alpha_{\mathcal{R}_m(1)} & \alpha_{\mathcal{R}_m(2)} &\cdots & \alpha_{\mathcal{R}_m(\rho_m)}\\ 
        \vdots & \vdots & \vdots & \vdots \\
        \alpha_{\mathcal{R}_m(1)}^{L-1} & \alpha_{\mathcal{R}_m(2)}^{L-1} & \cdots & \alpha_{\mathcal{R}_m(\rho_m)}^{L-1}
    \end{bmatrix},\\
    &\mathbf{D}_m^v={\rm diag}(v_{m,\mathcal{R}_m(1)}, v_{m,\mathcal{R}_m(2)} \cdots, v_{m,\mathcal{R}_m(\rho_m)}),\\
    &\mathbf{D}_m^u={\rm diag}(u_{m,1},u_{m,2},\ldots,u_{m,L}).
\end{align}
Then we can rewrite $V_1,V_2,\cdots,V_L$ in the following matrix form,
\begin{align}
    &[V_1,V_2,\ldots,V_L]^\mathsf{T}\notag\\
    =&\sum_{m\in[M]}
    \mathbf{V}_m
    \mathbf{D}_m^v
    \mathbf{C}_m
    \mathbf{D}_m^u
    \left[
        \mathbf{W}_{m,(1)}^\mathsf{T}\boldsymbol{\lambda}_{m,(1)},
        \mathbf{W}_{m,(2)}^\mathsf{T}\boldsymbol{\lambda}_{m,(2)},
        \cdots,
        \mathbf{W}_{m,(L)}^\mathsf{T}\boldsymbol{\lambda}_{m,(L)}
    \right]^\mathsf{T}.\label{eq:vmatrix1}
\end{align}
According to the Vandermonde decomposition of Cauchy matrices in Lemma \ref{lemma:vdc} in Appendix, for each $m\in[M]$, we have
\begin{align}
    \mathbf{C}_m
    =-(\mathbf{D}_m^v)^{-1}
    \mathbf{V}_m^{-1}
    \begin{bmatrix}
        1 & 1 & \cdots & 1\\
        f_1 & f_2 & \cdots & f_L\\
        \vdots & \vdots & \vdots & \vdots \\
        f_1^{L-1} & f_2^{L-1} & \cdots & f_L^{L-1}
    \end{bmatrix}
    (\mathbf{D}_m^u)^{-1}.
\end{align}

Therefore, \eqref{eq:vmatrix1} can be simplified as follows.
\begin{align} 
    &\left[V_1,V_2,\cdots,V_L\right]^\mathsf{T}\notag\\
    =&-\begin{bmatrix}
        1 & 1 & \cdots & 1\\
        f_1 & f_2 & \cdots & f_L\\
        \vdots & \vdots & \cdots & \vdots\\
        f_1^{L-1} & f_2^{L-1} & \cdots & f_L^{L-1}\\
    \end{bmatrix}
    \begin{bmatrix}
        \sum_{m\in[M]}\mathbf{W}_{m,(1)}^\mathsf{T}\boldsymbol{\lambda}_{m,(1)}\\
        \sum_{m\in[M]}\mathbf{W}_{m,(2)}^\mathsf{T}\boldsymbol{\lambda}_{m,(2)}\\
        \vdots\\
        \sum_{m\in[M]}\mathbf{W}_{m,(L)}^\mathsf{T}\boldsymbol{\lambda}_{m,(L)}
    \end{bmatrix}.\label{eq:ia2}
\end{align}
Because the constants $(f_\ell)_{\ell\in[L]}$ are distinct, the Vandermonde matrix on the RHS of \eqref{eq:ia2} has full rank, the desired linear combination $\lambda_{\boldsymbol{\Lambda}}(\mathcal{W})=\left[\sum_{m\in[M]}\mathbf{W}_{m,(\ell)}^\mathsf{T}\boldsymbol{\lambda}_{m,(\ell)}\right]_{(\ell)\in[L]}$ is resolvable by inverting the Vandermonde matrix. Therefore, the user can reconstruct the desired linear combination from the answers $(A^{[\boldsymbol{\Lambda}]}_n)_{n\in[N]}$. Note that a total of $L=\min_{m\in[M]}(\rho_m-X_m-T_m)$ desired $q$-ary symbols are retrieved from a total of $N$ downloaded $q$-ary symbols, the rate achieved is $R=\min_{m\in[M]}(\rho_m-X_m-T_m)/N$. This completes the proof of Lemma \ref{lemma:asymm}.

\subsection{Achievability Proof of Theorem \ref{thm:main}}\label{sec:proofmain}
Throughout this section, we will only consider non-degenerated settings where $\min_{m\in[M]}\rho_m-X-T >0$ otherwise the asymptotic capacity is zero. As discussed at the beginning of this section, the main idea of our achievability proof is to identify an augmented system (a corresponding Asymm-GXSTPLC setting) that can be used to construct asymptotic capacity-achieving scheme for any given GXSTPLC problem by merging servers. To this end, we first need a couple of quantities that are related to the solution to the linear programming in \eqref{eq:lp}. Note that the existence of a rational solution to the problem \eqref{eq:lp} is always guaranteed due to the corner point theorem. In other words, the solution of a linear programming can always be represented by a corner point of the feasible region, i.e., a vertex of the feasible region. Since the coefficients of the linear programming in \eqref{eq:lp} are integer-valued, all of the corner points must be rational-valued. Let us fix one of such solutions as $(D_1^*,D_2^*,\cdots,D_N^*)=(q_1/p_1, q_2/p_2, \cdots, q_N/p_N)$ where for all $n\in[N]$, $p_n$ and $q_n$ are co-prime. Besides, since vertices of the feasible region of \eqref{eq:lp} are bounded by the $N$-dimensional hypercube $[0,1]^N$, we have $p_n\geq q_n$ for all $n\in[N]$. Let us set $L=\lcm(p_1, p_2, \cdots, p_N)$, and define $\tau_n^* = L*q_n/p_n\in \mathbb{Z}_{\geq 0}$ for all $n\in[N]$. Denote the storage pattern of the GXSTPLC problem with $N$ servers, $M$ message sets, security level $X$ and privacy level $T$ as $\mathcal{R}=\{\mathcal{R}_1,\mathcal{R}_2,\cdots,\mathcal{R}_M\}$, the \emph{augmented system} has a total of $\overline{N}=\sum_{n\in[N]}\tau_n^*$ servers and $\overline{M}=M$ message sets. For notational simplicity, the $\overline{N}$ servers $(1,2,\cdots, \overline{N})$ are equivalently double-indexed as $((1,1),(1,2),\cdots,(1,\tau_1^*), (2,1), (2,2), \cdots, (2,\tau_2^*),\cdots, (N,1), (N,2), \cdots, (N,\tau_N^*))$ to reflect their hierarchical structures. Algorithm \ref{alg:genaugsys} generates the storage pattern $\overline{\mathcal{R}}=\{\overline{\mathcal{R}}_1,\overline{\mathcal{R}}_2,\cdots,\overline{\mathcal{R}}_M\}$ and asymmetric security and privacy thresholds $\overline{\mathbf{X}}=(\overline{X}_1, \overline{X}_2, \cdots, \overline{X}_M)$, $\overline{\mathbf{T}}=(\overline{T}_1, \overline{T}_2, \cdots, \overline{T}_M)$ for the augmented system.

\begin{algorithm}
\caption{Generation of the augmented system}\label{alg:genaugsys}
\KwIn{$N, M, L, (\tau_n^*)_{n\in[N]}, \mathcal{R}, X, T$}
\KwOut{$\overline{\mathcal{R}}, \overline{\mathbf{X}}, \overline{\mathbf{T}}$}
\ForEach{$m\in[M]$}{
    $\overline{\mathcal{R}}_m\gets \emptyset$\;
    $\overline{\boldsymbol{\tau}}^*_m\gets$ sorted $(\tau_n^*)_{n\in\mathcal{R}_m}$ in non-ascending order\;
    $\gamma_m\gets$ $(X+T+1)^{th}$ element in $\overline{\boldsymbol{\tau}}^*_m$\;
    $\overline{X}_m\gets X\gamma_m$\;
    $\overline{T}_m\gets T\gamma_m$\;
    \ForEach{$n\in\mathcal{R}_m$}{
        $\delta_{m,n}\gets \min(\gamma_m, \tau^*_n)$\;
        $\overline{\mathcal{R}}_m\gets \overline{\mathcal{R}}_m\cup \{(n,1),(n,2),\cdots,(n,\delta_{m,n})\}$\;
    }
}
\end{algorithm}
Upon the generation of the augmented system, the corresponding GXSTPLC scheme is obtained by merging servers $(n,1),(n,2),\cdots,(n,\tau_n^*)$ into Server $n$ for all $n\in[N]$, i.e., assigning the storage and queries of the augmented system for servers $(n,1),(n,2),\cdots,(n,\tau_n^*)$ to Server $n$. Note that this indeed recovers the storage pattern of the original GXSTPLC problem since it is easy to verify that for all $n\in[N]$, we have $\bigcup_{i\in[\tau_n^*]}\overline{\mathcal{M}}_{(n,i)}\subseteq\mathcal{M}_n$ where $\overline{\mathcal{M}}_{(n,i)}=\{m\in[M] \mid \overline{\mathcal{R}}_m\ni (n,i)\}$ is the dual representation of the storage pattern of the augmented system $\overline{\mathcal{R}}$. The following two lemmas conclude our achievability proof by showing that the GXSTPLC scheme constructed from the augmented system is $X$-secure and $T$-private, and more importantly, achieves the asymptotic capacity.
\begin{lemma}
    The GXSTPLC scheme constructed from the augmented system (generated by Algorithm \ref{alg:genaugsys}) satisfies $X$-secure and $T$-private constraints.
\end{lemma}
\begin{proof}
    Due to the symmetry between the definition of the structure of storage and queries, it suffices to prove $X$-security since $T$-privacy is similarly proved. Note that by the construction of the GXSTPLC scheme, for all $n\in[N]$,
    \begin{align}
        \mathcal{S}_n=\{\overline{\mathcal{S}}_{(n,1)},\overline{\mathcal{S}}_{(n,2)},\cdots,\overline{\mathcal{S}}_{(n,\tau^*_n)}\},
    \end{align}
    where $\overline{\mathcal{S}}_{(n,i)}$ denotes the storage at Server $(n,i)$ of the augmented system, $n\in[N], i\in[\tau_n^*]$. Recall the definition of the storage in \eqref{eq:defstor}, we can write
    \begin{align}
        \mathcal{S}_n=\left\{\widetilde{W}_{m,k}^{(n,i)} \mid i\in[\tau_n^*], m\in\overline{\mathcal{M}}_{(n,i)}, k\in[K_m]\right\}.
    \end{align}
    Therefore, for all $\mathcal{X}\subset[N], |\mathcal{X}|=X$,
    \begin{align}
        &I(\mathcal{S}_{\mathcal{X}}; \mathcal{W})\notag\\
        =&I\left(\left\{\widetilde{W}_{m,k}^{(n,i)} \mid n\in\mathcal{X}, i\in[\tau_n^*], m\in\overline{\mathcal{M}}_{(n,i)}, k\in[K_m]\right\}; \{W_{m,k}\}_{m\in[M], k\in[K_m]}\right)\\
        \overset{\eqref{eq:storind}}{=}&I\Big(\Big\{\widetilde{W}_{m,k}^{(n,i)} \mid n\in\mathcal{X}, i\in[\tau_n^*], \notag\\  
        & \hspace{2.5cm}m\in\overline{\mathcal{M}}_{(n,i)}, k\in[K_m]\Big\}; \Big\{W_{m,k}\mid m\in\bigcup_{n\in\mathcal{X}, i\in[\tau^*_n]}\overline{\mathcal{M}}_{(n,i)}, k\in[K_m]\Big\}\Big)\\
        \overset{\eqref{eq:storind}}{\leq}&\sum_{m\in\bigcup_{n\in\mathcal{X}, i\in[\tau^*_n]}\overline{\mathcal{M}}_{(n,i)}} I\left(\left\{\widetilde{W}_{m,k}^{(n,i)} \mid n\in\mathcal{X}, (n,i)\in\overline{\mathcal{R}}_m, k\in[K_m]\right\}; \left\{W_{m,k}\mid k\in[K_m]\right\}\right).\label{eq:proofsec1}
    \end{align}
    Recall that according to Algorithm \ref{alg:genaugsys}, for all $m\in[M]$, the amount of ordered pairs that have the form $(n,*)$ in the set $\overline{\mathcal{R}}_m$ is $\delta_{m,n}\leq \gamma_m$. Therefore, each summand in \eqref{eq:proofsec1} equals zero since $\overline{X}_m=X\gamma_m$. This completes the proof.
\end{proof}
\begin{lemma}
    The rate $R=C_{\infty}(\mathcal{R})$ is achievable for the problem of GXSTPLC.
\end{lemma}
\begin{proof}
    It suffices to show that the rate $R=C_{\infty}(\mathcal{R})$ is achievable for the augmented system since the GXSTPLC scheme constructed from the augmented system achieves the same rate. According to Algorithm \ref{alg:genaugsys}, it is easy to see that $\overline{\rho}_m=|\overline{\mathcal{R}}_m|=(X+T)\gamma_m+\nu_m$ where $\nu_m$ is the summation of the smallest $(\rho_m-X-T)$ elements in $\{\tau^*_n\}_{n\in\mathcal{R}_m}$. Note that by the definition of $\tau_n^*, n\in[N]$, for arbitrary set $\mathcal{R}'_m\subset\mathcal{R}_m$ such that $|\mathcal{R}'_m|=\rho_m-X-T$, we have $\sum_{n\in\mathcal{R}'_m}\tau_n^*=L\sum_{n\in\mathcal{R}'_m}q_n/p_n\geq L$ since $(q_1/p_1, q_2/p_2, \cdots, q_n/p_n)$ must lie in the feasible region of the linear programming \eqref{eq:lp}. Therefore, we conclude that $\nu_m\geq L$ for all $m\in[M]$, and hence $\min_{m\in[M]}(\overline{\rho}_m-\overline{X}_m-\overline{T}_m)= \min_{m\in[M]}\nu_m\geq L$. Revoking Lemma \ref{lemma:asymm} completes the proof since the rate $R=L/\overline{N}=L/\sum_{n\in{N}}\tau^*_n=\left(\sum_{n\in[N]}(q_n/p_n)\right)^{-1}=C_{\infty}(\mathcal{R})$ is achievable.
\end{proof}
Let us further elaborate via the following examples.
\subsubsection{Example 1}
Consider a motivating example where we have $N=6$ servers and $K$ messages, $X=1$ and $T=1$. The $K$ messages are partitioned into two disjoint sets $\mathcal{W}_1$, and $\mathcal{W}_2$. The storage pattern for this example is as follows.
\begin{subequations}
    \begin{align}
        \mathcal{R}_1&=\{1,2,3,5\},\\
        \mathcal{R}_2&=\{3,4,6\}.
    \end{align}
\end{subequations}
According to Theorem \ref{thm:main}, we have the following bounds $D_3\geq 1$, $D_4\geq 1$, $D_6\geq 1$, $D_1+D_2\geq 1$, $D_2+D_5\geq 1$, $D_1+D_5\geq 1$. Therefore, $\sum_{n\in[6]}D_n=D_3+D_4+D_6+\frac{1}{2}(D_1+D_2+D_2+D_5+D_1+D_5)\geq 9/2$, where the equality holds when $(D_1,D_2,\cdots,D_6)=(1/2, 1/2, 1, 1, 1/2, 1)$. Following the notations above, let us set $(\tau_1^*, \tau_2^*, \cdots, \tau_6^*)=(1,1,2,2,1,2)$, $L=2$ and $\overline{N}=9$. The $9$ servers in the augmented system are listed as $((1,1),(2,1), (3,1),(3,2), (4,1),(4,2), (5,1), (6,1), (6,2))$.

Now let us generate the augmented system according to Algorithm \ref{alg:genaugsys}. Specifically, for the augmented system we have $\overline{X}_1=\overline{T}_1=1$ and $\overline{X}_2=\overline{T}_2=2$, and the storage pattern is shown in Table \ref{table:ex1}.

\begin{table}[h]
    \caption{Storage pattern of the augmented system in Example 1}
    \begin{center}
        \begin{tabular}{cccccccccc}%
            \hline \textbf{Server}&$(1,1)$&$(2,1)$&$(3,1)$&$(3,2)$&$(4,1)$&$(4,2)$&$(5,1)$&$(6,1)$&$(6,2)$\\\hline 
            \multirow{2}{*}{$\overline{\mathcal{M}}_{(n,i)}$}&$1$&$1$&$1$&$2$&$2$&$2$&$1$&$2$&$2$\\
             & & &$2$& & & & & \\\hline 
        \end{tabular}
        \label{table:ex1}
    \end{center}
\end{table}
This is exactly the setting presented in the motivating example in Section \ref{sec:asymmex2} and the rate $R=2/9$ is achievable for the augmented system which matches the converse bound of the corresponding GXSTPLC setting. Now let us see why is the resulting GXSTPLC scheme $X=1$-secure and $T=1$-private. Note that
\begin{align}
    I(\mathcal{S}_1; \mathcal{W})=&I\left(\left\{\widetilde{W}_{1,k}^{(1,1)}\right\}_{k\in[K_1]}; \{W_{1,k}\}_{k\in[K_m]}\right)=0,\label{eq:exxs1}\\
    I(Q^{[\boldsymbol{\Lambda}]}_1; \boldsymbol{\Lambda})
    =&I\left(
    \left\{\widetilde{\lambda}_{1,k}^{(1,1)}\right\}_{k\in[K_1]}; 
    \left\{\lambda_{1,k,(\ell)}\right\}_{k\in[K_1],\ell\in[2]}
    \right)=0,\\
    I(\mathcal{S}_2; \mathcal{W})
    =&I\left(
    \left\{\widetilde{W}_{1,k}^{(2,1)}\right\}_{k\in[K_1]}; 
    \{W_{1,k}\}_{k\in[K_1]}
    \right)=0,\\
    I(Q^{[\boldsymbol{\Lambda}]}_2; \boldsymbol{\Lambda})
    =&I\left(
    \left\{\widetilde{\lambda}_{1,k}^{(2,1)}\right\}_{k\in[K_1]}; 
    \{\lambda_{1,k,(\ell)}\}_{k\in[K_1],\ell\in[2]}
    \right)=0,\\
    I(\mathcal{S}_3; \mathcal{W})
    =&I\left(
    \left\{\widetilde{W}_{1,k}^{(3,1)}\right\}_{k\in[K_1]},
    \left\{\widetilde{W}_{2,k}^{(3,1)},\widetilde{W}_{2,k}^{(3,2)}\right\}_{k\in[K_2]}; 
    \{W_{m,k}\}_{m\in[2],k\in[K_m]}
    \right)\\
    \le & I\left(
    \left\{\widetilde{W}_{1,k}^{(3,1)}\right\}_{k\in[K_1]};
    \{W_{1,k}\}_{k\in[K_1]}\right)\notag\\
    &+
    I\left(
    \left\{\widetilde{W}_{2,k}^{(3,1)},\widetilde{W}_{2,k}^{(3,2)}\right\}_{k\in[K_2]}; 
    \{W_{2,k}\}_{k\in[K_2]}
    \right)=0,\\
    I(Q^{[\boldsymbol{\Lambda}]}_3; \boldsymbol{\Lambda})
    =&I\left(
    \left\{\widetilde{\lambda}_{1,k}^{(3,1)}\right\}_{k\in[K_1]},
    \left\{\widetilde{\lambda}_{2,k}^{(3,1)},\widetilde{\lambda}_{2,k}^{(3,2)}\right\}_{k\in[K_2]}; 
    \{\lambda_{m,k,(\ell)}\}_{m\in[2],k\in[K_m],\ell\in[2]}
    \right)\\
    \le &I\left(
    \left\{\widetilde{\lambda}_{1,k}^{(3,1)}\right\}_{k\in[K_1]};
    \{\lambda_{1,k,(\ell)}\}_{k\in[K_1],\ell\in[2]}\right)\notag\\
    &+
    I\left(
    \left\{\widetilde{\lambda}_{2,k}^{(3,1)},\widetilde{\lambda}_{2,k}^{(3,2)}\right\}_{k\in[K_2]}; 
    \{\lambda_{2,k,(\ell)}\}_{k\in[K_2],\ell\in[2]}
    \right)=0,\\
    I(\mathcal{S}_4; \mathcal{W})
    =&I\left(
    \left\{\widetilde{W}_{2,k}^{(4,1)},\widetilde{W}_{2,k}^{(4,2)}\right\}_{k\in[K_2]};
    \{W_{2,k}\}_{k\in[K_2]}
    \right)=0,\\
    I(Q^{[\boldsymbol{\Lambda}]}_3; \boldsymbol{\Lambda})
    =&I\left(
    \left\{\widetilde{\lambda}_{2,k}^{(4,1)},\widetilde{\lambda}_{2,k}^{(4,2)}\right\}_{k\in[K_2]}; 
    \{\lambda_{2,k,(\ell)}\}_{k\in[K_2],\ell\in[2]}
    \right)=0,\\
    I(\mathcal{S}_5; \mathcal{W})=&I\left(\left\{\widetilde{W}_{1,k}^{(5,1)}\right\}_{k\in[K_1]}; \{W_{1,k}\}_{k\in[K_1]}\right)=0,\\
    I(Q^{[\boldsymbol{\Lambda}]}_5; \boldsymbol{\Lambda})
    =&I\left(
    \left\{\widetilde{\lambda}_{1,k}^{(5,1)}\right\}_{k\in[K_1]}; 
    \left\{\lambda_{1,k,(\ell)}\right\}_{k\in[K_1],\ell\in[2]}
    \right)=0,\\
    I(\mathcal{S}_6; \mathcal{W})
    =&I\left(
    \left\{\widetilde{W}_{2,k}^{(6,1)},\widetilde{W}_{2,k}^{(6,2)}\right\}_{k\in[K_2]};
    \{W_{2,k}\}_{k\in[K_2]}
    \right)=0,\\
    I(Q^{[\boldsymbol{\Lambda}]}_3; \boldsymbol{\Lambda})
    =&I\left(
    \left\{\widetilde{\lambda}_{2,k}^{(6,1)},\widetilde{\lambda}_{2,k}^{(6,2)}\right\}_{k\in[K_2]}; 
    \{\lambda_{2,k,(\ell)}\}_{k\in[K_2],\ell\in[2]}
    \right)=0.
\end{align}
where the last equalities in the above equations hold due to $\overline{X}_1=\overline{T}_1=1$ and $\overline{X}_2=\overline{T}_2=2$ for the augmented system. Hence the resulting GXSTPLC scheme $X=1$-secure and $T=1$-private.

\subsubsection{Example 2}
Consider a (more sophisticated) example where we have $N=14$ servers and $K$ messages, and similarly, let us set $X=1$ and $T=1$. The $K$ messages are partitioned into four disjoint sets $\mathcal{W}_1$, $\mathcal{W}_2$, $\mathcal{W}_3$, and $\mathcal{W}_4$. The storage pattern for this example is as follows.
\begin{subequations}
    \begin{align}
        \mathcal{R}_1=&\{1,4,7,9\},\\
        \mathcal{R}_2=&\{1,3,4,5,8\},\\
        \mathcal{R}_3=&\{3,4,6,8,10,13\},\\
        \mathcal{R}_4=&\{2,6,10,11,12,13,14\}.
    \end{align}
\end{subequations}
Note that the following bounds hold due to Theorem \ref{thm:main}.
\begin{subequations}
    \begin{align}
        D_1+D_4\geq& 1,~ D_7+D_9\geq 1,~ D_3+D_5+D_8\geq 1,\\
        \sum_{n\in\mathcal{R}'}D_n\geq& 1, ~\forall \mathcal{R}'\subset\mathcal{R}_4, |\mathcal{R}'|=5.
    \end{align}
\end{subequations}
Therefore, $\sum_{n\in[14]}D_n=D_1+D_4+D_7+D_9+D_3+D_5+D_8+\frac{1}{15}\sum_{\mathcal{R}'\subset\mathcal{R}_4, |\mathcal{R}'|=5}\sum_{n\in\mathcal{R}'}D_n\geq 44/10$, where the equality holds when $(D_1,D_2,\cdots,D_{14})=(1/2, 1/5, 2/5, 1/2, 1/10, 1/5, 1/2, 1/2, 1/2, 1/5, 1/5, 1/5, 1/5, 1/5)$. Following the notations above, let us set $(\tau_1^*, \tau_2^*, \cdots, \tau_{14}^*)=(5, 2, 4, 5, 1, 2, 5, 5, 5, 2, 2, 2, 2, 2)$, $L=10$ and $\overline{N}=44$. The $44$ servers in the augmented system are listed as $((1,1),(1,2),(1,3),(1,4),(1,5),\cdots,(14,1),(14,2))$. According to Algorithm \ref{alg:genaugsys}, the generated augmented system is as follows: $\overline{X}_1=\overline{T}_1=5$, $\overline{X}_2=\overline{T}_2=5$, $\overline{X}_3=\overline{T}_3=4$, $\overline{X}_4=\overline{T}_4=2$, and the storage pattern is shown in Table \ref{tb:ex2}.
\begin{table}[!h]
    \caption{Storage pattern of the augmented system in Example 2}
    \begin{center}
        \begin{tabular}{ccccccccccc}%
            \hline \textbf{Server}&
            (1,1)&(1,2)&(1,3)&(1,4)&(1,5)&(2,1)&(2,2)&(3,1)&(3,2)&(3,3)\\\hline 
            \multirow{2}{*}{$\overline{\mathcal{M}}_{(n,i)}$}&
             1&1&1&1&1&4&4&2&2&2\\
            &2&2&2&2&2& & &3&3&3\\
            \hline \textbf{Server}&
            (3,4)&(4,1)&(4,2)&(4,3)&(4,4)&(4,5)&(5,1)&(6,1)&(6,2)&(7,1)\\\hline 
            \multirow{3}{*}{$\overline{\mathcal{M}}_{(n,i)}$}& 
             2&1&1&1&1&1&2&3&3&1\\
            &3&2&2&2&2&2& &4&4&\\
            & &3&3&3&3& & & & &\\
            \hline \textbf{Server}&
            (7,2)&(7,3)&(7,4)&(7,5)&(8,1)&(8,2)&(8,3)&(8,4)&(8,5)&(9,1)\\\hline 
            \multirow{2}{*}{$\overline{\mathcal{M}}_{(n,i)}$}&
             1&1&1&1&2&2&2&2&2&1\\
            & & & & &3&3&3&3& &\\
            \hline \textbf{Server}&
            (9,2)&(9,3)&(9,4)&(9,5)&(10,1)&(10,2)&(11,1)&(11,2)&(12,1)&(12,2)\\\hline 
            \multirow{2}{*}{$\overline{\mathcal{M}}_{(n,i)}$}&
             1&1&1&1&3&3&4&4&4&4\\
            & & & & &4&4& & & & \\
            \hline \textbf{Server}&
            (13,1)&(13,2)&(14,1)&(14,2)&&&&&&\\\hline
            \multirow{2}{*}{$\overline{\mathcal{M}}_{(n,i)}$}&
             3&3&4&4&&&&&&\\
            &4&4& & &&&&&& \\\hline
        \end{tabular}
        \label{tb:ex2}
    \end{center}
\end{table}
According to Lemma \ref{lemma:asymm}, due to the fact that $\overline{\rho}_1=|\overline{\mathcal{R}}_1|=20$, $\overline{\rho}_2=|\overline{\mathcal{R}}_2|=20$, $\overline{\rho}_3=|\overline{\mathcal{R}}_3|=18$, $\overline{\rho}_4=|\overline{\mathcal{R}}_4|=14$, the rate $R=10/44=5/22$ is achievable for the augmented system, which matches the converse bound of the corresponding GXSTPLC setting. Similarly, it is straightforward to check that the GXSTPLC scheme is $X=1$-secure and $T=1$-private. For example, consider the first server, we have
\begin{align}
    I(\mathcal{S}_1; \mathcal{W})
    =&I\Big(
    \left\{\widetilde{W}_{1,k}^{(1,1)},\widetilde{W}_{1,k}^{(1,2)},\widetilde{W}_{1,k}^{(1,3)},\widetilde{W}_{1,k}^{(1,4)},\widetilde{W}_{1,k}^{(1,5)}\right\}_{k\in[K_1]},\notag\\
    &\hspace{1cm}\left\{\widetilde{W}_{2,k}^{(1,1)},\widetilde{W}_{2,k}^{(1,2)},\widetilde{W}_{2,k}^{(1,3)}\widetilde{W}_{2,k}^{(1,4)}\widetilde{W}_{2,k}^{(1,5)}\right\}_{k\in[K_2]}; 
    \{W_{m,k}\}_{m\in[2],k\in[K_m]}
    \Big)\\
    \le & I\left(
    \left\{\widetilde{W}_{1,k}^{(1,1)},\widetilde{W}_{1,k}^{(1,2)},\widetilde{W}_{1,k}^{(1,3)},\widetilde{W}_{1,k}^{(1,4)},\widetilde{W}_{1,k}^{(1,5)}\right\}_{k\in[K_1]};
    \{W_{1,k}\}_{k\in[K_1]}\right)\notag\\
    &+
    I\left(
    \left\{\widetilde{W}_{2,k}^{(1,1)},\widetilde{W}_{2,k}^{(1,2)},\widetilde{W}_{2,k}^{(1,3)}\widetilde{W}_{2,k}^{(1,4)}\widetilde{W}_{2,k}^{(1,5)}\right\}_{k\in[K_2]}; 
    \{W_{2,k}\}_{k\in[K_2]}
    \right)=0,\\
    I(Q^{[\boldsymbol{\Lambda}]}_1; \boldsymbol{\Lambda})
    =&I\left(
    \left\{\widetilde{\lambda}_{m,k}^{(1,1)},\widetilde{\lambda}_{m,k}^{(1,2)},\widetilde{\lambda}_{m,k}^{(1,3)},\widetilde{\lambda}_{m,k}^{(1,4)},\widetilde{\lambda}_{m,k}^{(1,5)}\right\}_{m\in[2],k\in[K_m]}; 
    \{\lambda_{m,k,(\ell)}\}_{m\in[2],k\in[K_m],\ell\in[10]}
    \right)\\
    \le &I\left(
    \left\{\widetilde{\lambda}_{1,k}^{(1,1)},\widetilde{\lambda}_{1,k}^{(1,2)},\widetilde{\lambda}_{1,k}^{(1,3)},\widetilde{\lambda}_{1,k}^{(1,4)},\widetilde{\lambda}_{1,k}^{(1,5)}\right\}_{k\in[K_1]};
    \{\lambda_{1,k,(\ell)}\}_{k\in[K_1],\ell\in[10]}\right)\notag\\
    &+
    I\left(
    \left\{\widetilde{\lambda}_{2,k}^{(1,1)},\widetilde{\lambda}_{2,k}^{(1,2)},\widetilde{\lambda}_{2,k}^{(1,3)},\widetilde{\lambda}_{2,k}^{(1,4)},\widetilde{\lambda}_{2,k}^{(1,5)}\right\}_{k\in[K_2]}; 
    \{\lambda_{2,k,(\ell)}\}_{k\in[K_2],\ell\in[10]}
    \right)=0.
\end{align}
where the last equalities in the above two equations hold due to the fact that $\overline{X}_1=\overline{T}_1=\overline{X}_2=\overline{T}_2=5$ for the augmented system. The rest independencies can be similarly verified.

\begin{remark}\label{remark:gxslc}
    Our result also settles the exact (i.e., non-asymptotic) capacity of GXSLC, namely, the special setting of GXSTPLC where $T=0, X\geq 1$ and $\lambda_{m,k,(\ell)}\neq 0$ for all $m\in[M], k\in[K_m], \ell\in[L]$. Specifically, the capacity is the same as Theorem \ref{thm:main} (set $T=0$) even for finite $K$. It is obvious that the achievability scheme follows directly by setting $T=0$, therefore let us briefly present the converse argument as the converse for Theorem \ref{thm:main} applies only to $T\geq 1$. For any GXSLC scheme and arbitrary $m\in[M]$, it is possible to set $W_{m',k}=\textbf{0}$ for all $m'\in[M]\setminus\{m\}, k\in[K_{m'}]$ and let $W_{m,k}$ be i.i.d. uniform. Besides, let us give $\{\widetilde{W}_{m',k}^{(n)}\}_{m'\in[M]\setminus\{m\}, k\in[K_{m'}], n\in\mathcal{R}_{m'}}$ to the user for free. Denoting the expected number of downloaded $q$-ary symbols from Server $n$ of the GXSLC scheme as $D_n$, and the expected number of downloaded $q$-ary symbols from Server $n$ of the resulting scheme as $D'_n$, we have $D'_n=D_n$ for all $n\in\mathcal{R}_m$ and $D'_n=0$ otherwise. This is because conditioning on messages taking values over a particular subset of the original alphabet cannot change $D_{[N]}$ since we have non-trivial security constraint $X\geq 1$. Besides, as shares of messages $\{W_{m',k}\}_{m'\in[M]\setminus\{m\}, k\in[K_{m'}]}$ are given to the user for free, it is not necessary to download anything from Server $n, n\in[N]\setminus\mathcal{R}_{m}$ as the storage at these servers is now available at the user. Note that the resulting scheme is actually an $X$-secure linear computation scheme (without graph based replicated storage) for which the following standard secret sharing bounds must hold (see, e.g., \cite{Jia_Jafar_SDBMM}).
    \begin{align}
        \sum_{n\in\mathcal{R}'_m} D'_n = \sum_{n\in\mathcal{R}'_m} D_n \geq 1, ~~\forall \mathcal{R}'_{m}\subset\mathcal{R}_m, |\mathcal{R}'_{m}|=|\mathcal{R}_m|-X.
    \end{align}
    The desired converse bound is hence obtained by applying the above argument to all $m\in[M]$.
\end{remark}

\section{Conclusion}\label{sec:conclusion}
We explored the problem of $X$-secure $T$-private linear computation with graph based replicated storage (GXSTPLC) and completely characterized its asymptotic capacity, which also settles a previous conjecture on the asymptotic capacity of $X$-secure $T$-private information retrieval with graph based replicated storage (GXSTPIR) in \cite{Jia_Jafar_GXSTPIR}. The result is timely due to emerging distributed storage/computation applications where security, privacy and heterogeneous data availability concerns come into play. From a theoretical standpoint, our result reveals the interesting interactions among these competing factors under the umbrella of private linear computation/information retrieval that may lead to a profound understanding of a set of fundamental problems such as private information retrieval/computation, secret sharing, graph based replicated storage as well as coded storage. While our result completely settles the asymptotic capacity of GXSTPLC, it is worth noting that our scheme indeed requires extra storage overhead (compared to the amount of storage required by simply replicating the messages according to the storage pattern) when $X\geq 1$\footnote{When $X=0$, our scheme does not require extra storage overhead even if we construct the GXSTPLC scheme by merging servers in the augmented system. This is because when $X=0$, the messages are directly replicated at the corresponding servers without any secret sharing/coding, hence no more storage is necessary beyond trivial replications.}. Whether or not this extra storage overhead is necessary for asymptotic capacity achieving schemes is perhaps one of the most important open problems for future work which calls for potentially new GXSTPLC schemes, storage/query designs and/or converse arguments. Also, it is promising to explore applications of ideas in this work to non-asymptotic settings.

\section*{Acknowledgments}
The authors would like to thank Professor Syed A. Jafar at UC Irvine for his insightful comments in relation to this work.

\bibliography{Thesis}
\bibliographystyle{IEEEtran}

\appendix

The following two important lemmas are used in the proof of Lemma \ref{lemma:asymm}. In particular, Lemma \ref{lemma:grs} presents a structure inspired by dual GRS codes due to Jia et al. in \cite{Jia_Jafar_GXSTPIR}, and Lemma \ref{lemma:vdc} shows a Vandermonde decomposition of Cauchy matrices due to Gow in \cite{gow1992cauchy}.
\begin{lemma}\label{lemma:grs}
(\cite[Lemma 2]{Jia_Jafar_GXSTPIR})
Let $\alpha_1,\cdots,\alpha_n$ be distinct elements from a finite field $\mathbb{F}_q$. Define $v_1,v_2,\cdots,v_n$ as follows.
\begin{align}
    v_i=\prod_{j\in[N]\setminus\{i\}}\left(\alpha_i-\alpha_j\right)^{-1},\  i\in[n].\label{eq:defv}
\end{align}
Then the following identity is satisfied.
\begin{align}
    \sum_{i\in[n]}v_i\alpha^j_i=0,\ \forall j\in\{0,1,\ldots,n-2\}.
\end{align}
\end{lemma}

\begin{lemma}\label{lemma:vdc}(\cite[Theorem 1]{gow1992cauchy})
    Let $\alpha_1,\alpha_2,\cdots,\alpha_n$ and $f_1,f_2,\cdots,f_\ell$ be distinct elements from a finite field $\mathbb{F}_q$, $n\ge \ell$. Define the following matrices.
    \begin{align}
        \mathbf{C}&=
        \begin{bmatrix}
            \frac{1}{\alpha_1-f_1} & \frac{1}{\alpha_1-f_2} & \cdots & \frac{1}{\alpha_1-f_\ell}\\
            \frac{1}{\alpha_2-f_1} & \frac{1}{\alpha_2-f_2} & \cdots & \frac{1}{\alpha_2-f_\ell}\\
            \vdots & \vdots & \cdots & \vdots\\
            \frac{1}{\alpha_n-f_1} & \frac{1}{\alpha_n-f_2} & \cdots & \frac{1}{\alpha_n-f_\ell}
        \end{bmatrix},&
        \mathbf{V}_\alpha&=
        \begin{bmatrix}
            1 & 1 & \cdots & 1\\
            \alpha_1 & \alpha_2 & \cdots & \alpha_n\\
            \vdots & \vdots & \cdots & \vdots\\
            \alpha_1^{n-1} & \alpha_2^{n-1} & \cdots & \alpha_n^{n-1}\\
        \end{bmatrix},\\
        \mathbf{V}_f&=
        \begin{bmatrix}
            1 & 1 & \cdots & 1\\
            f_1 & f_2 & \cdots & f_\ell\\
            \vdots & \vdots & \cdots & \vdots\\
            f_1^{n-1} & f_2^{n-1} & \cdots & f_\ell^{n-1}\\
        \end{bmatrix},\\
        \mathbf{D}_v&=\textrm{diag}(v_1,v_2,\cdots,v_n),\\ 
        \mathbf{D}_u&=\textrm{diag}(u_1,u_2,\cdots,u_\ell),
    \end{align}
    where $v_1,v_2,\cdots,v_n$ and $u_1,u_2,\cdots,u_\ell$ are defined as follows.
    \begin{align}
        v_i=&\prod_{k\in[n]\setminus\{i\}}\left(\alpha_i-\alpha_k\right),\ i\in[n]\\
        u_j=&\prod_{i\in[n]}\left(f_j-\alpha_i\right),\ j\in[\ell].
    \end{align}
    Then we have the following decomposition of the Cauchy matrix $\mathbf{C}$ in terms of (scaled) Vandermonde matrices.
    \begin{align}
        \mathbf{C}=-\mathbf{D}_v\mathbf{V}_\alpha^{-1}\mathbf{V}_f\mathbf{D}_u^{-1}.
    \end{align}
\end{lemma}
\begin{proof}
    
    While the lemma is proved in \cite{gow1992cauchy}, here we present an alternative shorter proof. Let us define the following polynomials
    \begin{align}
        p_i(x)=\frac{\prod_{k\in[n]}(x-\alpha_k)}{x-\alpha_i}=\sum_{j\in[n]}a_{ij}x^{j-1}, ~\forall i\in[n], \label{eq:piBasis}
    \end{align}
    as well as the coefficient matrix 
    \begin{align}
        \mathbf{A}=\begin{bmatrix}
            a_{11} & a_{12} & \cdots & a_{1n}\\ 
            a_{21} & a_{22} & \cdots & a_{2n}\\ 
            \vdots & \vdots & \cdots & \vdots\\
            a_{n1} & a_{n2} & \cdots & a_{nn}
        \end{bmatrix}.
    \end{align}
    It is hence easy to check that the following two identities hold.
    \begin{align}
        \mathbf{D}_v=&
        \begin{bmatrix}
            p_1(\alpha_1) & p_1(\alpha_2) & \cdots & p_1(\alpha_n)\\
            p_2(\alpha_1) & p_2(\alpha_2) & \cdots & p_2(\alpha_n)\\
            \vdots & \vdots & \cdots & \vdots\\
            p_n(\alpha_1) & p_n(\alpha_2) & \cdots & p_n(\alpha_n)
        \end{bmatrix}=\mathbf{A}\mathbf{V}_\alpha,\\
        -\mathbf{C}\mathbf{D}_u=&
        \begin{bmatrix}
            p_1(f_1) & p_1(f_2) & \cdots & p_1(f_\ell)\\
            p_2(f_1) & p_2(f_2) & \cdots & p_2(f_\ell)\\
            \vdots & \vdots & \cdots & \vdots\\
            p_n(f_1) & p_n(f_2) & \cdots & p_n(f_\ell)
        \end{bmatrix}=\mathbf{A}\mathbf{V}_f.
    \end{align}
    Therefore, since $\mathbf{V}_{\alpha}$ (a Vandermonde matrix with distinct nodes) and $\mathbf{D}_u$ (a diagonal matrix with non-zero entries) are invertible, we have
    \begin{align}
        -\mathbf{C}\mathbf{D}_u&=\mathbf{D}_v\mathbf{V}_{\alpha}^{-1}\mathbf{V}_f\\
        \Rightarrow \mathbf{C}&=- \mathbf{D}_v\mathbf{V}_\alpha^{-1}\mathbf{V}_f \mathbf{D}_u^{-1}.
    \end{align}
    This completes the proof of Lemma \ref{lemma:vdc}.
    \end{proof}

\end{document}